\documentclass[11pt]{article}
\usepackage[margin=1in]{geometry}
\usepackage[T1]{fontenc}
\usepackage{lmodern}
\usepackage{microtype}
\usepackage{amsmath,amssymb,amsthm}
\usepackage{graphicx}
\graphicspath{{figures/}}
\usepackage{booktabs,array,multirow,longtable}
\usepackage[font=small,labelfont=bf]{caption}
\usepackage{url}
\usepackage[hidelinks]{hyperref}
\usepackage{enumitem}
\setlist{nosep,leftmargin=*}
\newtheorem{lemma}{Lemma}
\newtheorem{proposition}{Proposition}
\theoremstyle{definition}
\newtheorem{definition}{Definition}
\newcommand{\Grant}{\mathrm{Grant}}
\newcommand{\Sg}{S_{\mathrm{unsafe}}}
\newcommand{\Sres}{S_{\mathrm{res}}}
\newcommand{\Hreal}{H_{\mathrm{real}}}
\newcommand{\Habs}{H_{\mathrm{abs}}}
\newcommand{\epsm}{\varepsilon^{-}_{H}}
\newcommand{\epsp}{\varepsilon^{+}_{H}}
\newcommand{\rC}{r_{C}}
\newcommand{\lC}{\ell_{C}}
\newcommand{\hardDeny}{\mathrm{hardDeny}}

\begin{document}

\title{Auditing the Hand-Written Inputs of Learned  Access Control}
\author{Arsalan Jalil Sarghaleh\thanks{This work has been submitted to the IEEE for possible publication. Copyright may be transferred without notice, after which this version may no longer be accessible.}\\
Department of Electrical Engineering\\ Sharif University of Technology, Tehran, Iran\\
\texttt{arsalan.jalil@ee.sharif.edu}
\and
Mehdi Hazrati\\
Department of Electrical and Computer Engineering\\ University of Victoria, Victoria, BC, Canada\\
\texttt{smhazrati@uvic.ca}}
\date{}
\maketitle

\begin{abstract}
Learned access control is typically deployed behind two hand-built artefacts:
a \emph{hard-deny set} enforced by a small trusted shield, and a \emph{region
map} that decides which requests mined rules answer and which the learner
answers. Verifying the shield or measuring the learner validates neither. We
audit both on the five public ABAC Lab policies. First, we formalise hard-deny
soundness and check it exhaustively over relations of up to 794{,}250 tuples. Clauses that faithfully transcribe an
explicit restriction are unsound at rates of 0.50--0.70 across four
elicitors---an author and three blind language-model elicitors---because prose
scopes restrictions to a section while the grant relation is a disjunction over
all rules: 97\% of counterexamples are granted by other rules. Shipped
unchecked, these deny sets would withhold 16--78\% of everything the policies
grant; certified, they withhold nothing. Vacuity, a second silent failure, also
occurs. Second, hybrid routing needs a two-class,
abstaining rule set that ABAC miners do not emit, and ranking a
mixed-effect candidate pool under a cap is one-class in all 24 cells for the
two commonest ranking keys. Third, for additive losses the
hybrid-versus-learner effect factorises into covered mass times a rule--learner
risk contrast that must be reported on both false-non-grant and false-permit
loss: on two real policies, most selection rules give opposite verdicts on the
two losses, and the one favourable on both holds for one policy, one learner
configuration and a band of noise levels. A released artefact regenerates and
checks every reported number.
\end{abstract}


\section{Introduction}
\label{sec:intro}
Attribute-based access control evaluates rules against subject,
object, operation and environment attributes~\cite{hu2014abac}, and in practice
the governing policy is increasingly mined from logs rather than hand-written
\cite{xu2015mining,cotrini2018sparse,karimi2022automatic,shang2024hierarchical}.
Mined rules are approximations---no mining algorithm produces a completely
accurate policy~\cite{guo2025onthefly}---so they leave requests on which the
rules are silent or self-contradictory. A learned classifier can resolve exactly
those requests, and learned components are increasingly placed on the decision
path~\cite{aboukadri2024survey,nobi2025sok,nobi2022deep,cappelletti2019quality}. Because a learner's
failures are not intrinsically bounded---miscalibration~\cite{guo2017calibration},
distribution shift, poisoning and evasion~\cite{biggio2018wild}, and demonstrated
adversarial vulnerability of ML-based access control
specifically~\cite{llamas2025adversarial}---the standard architectural response
is a reference monitor~\cite{anderson1972} in its modern form: a small,
correct-by-construction runtime shield
\cite{schneider2000enforceable,ligatti2009runtime,bloem2015shield,alshiekh2018safe}
interposed before the actuator, with Simplex-style runtime assurance as
antecedent~\cite{sha2001simplicity}.

This paper is about what that architecture does \emph{not} contain. The shield
is verifiable and the learner is measurable, but the composed system also
depends on two artefacts that are neither:

\begin{itemize}
\item a \textbf{hard-deny set} $H$---a list of clauses, written by a person
reading the policy, that the shield refuses on unconditionally; and
\item a \textbf{region map}---the partition of the request space into requests
the rules answer and requests the learner answers, produced by a mining pipeline
whose selection step is a design choice.
\end{itemize}

Verifying the shield says nothing about whether $H$ is correct: a verified
authorizer over a wrong deny set is a correct authorizer producing wrong
denials. Measuring the learner says nothing about whether the region map is real:
if the mining pipeline's selection step is degenerate, the map is an artefact and
every downstream measurement inherits it. Both artefacts are hand-built, both sit
inside the loop, and neither is checked by anything else in the stack.

\begin{figure}[t]
\centering
\includegraphics[width=0.85\linewidth]{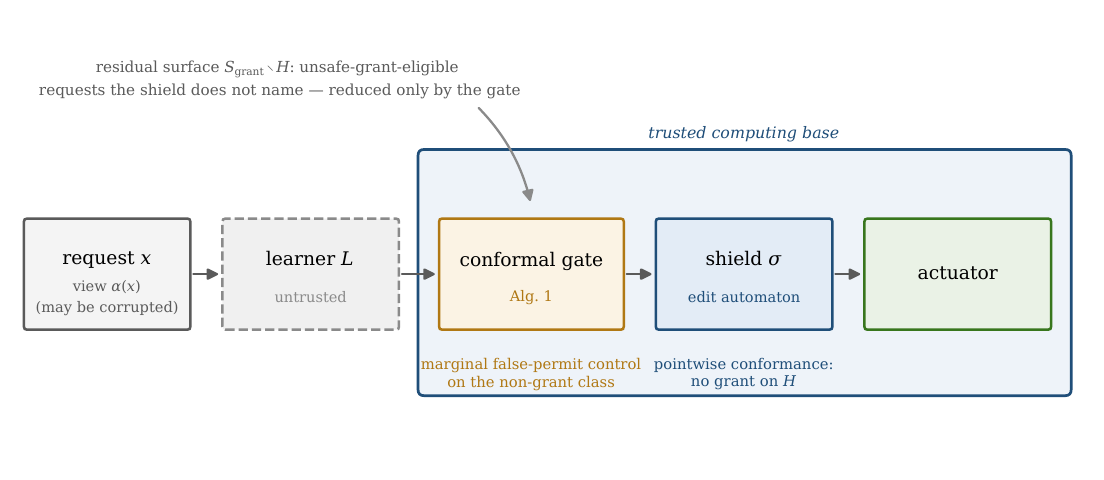}
\caption{The composed architecture, with the two hand-built artefacts this paper
audits marked. An untrusted learner proposes an action; a deterministic shield
admits no grant on the named hard-deny set $H$ and is the identity outside it.
The shield's behaviour on $H$ is a refinement property and is not in doubt
(Lemma~\ref{lem:conform}). What is in doubt is \emph{which requests are in $H$}
(Section~\ref{sec:c1}) and \emph{which requests reach the learner at all}
(Sections~\ref{sec:c2}--\ref{sec:c3}). Neither is produced or checked by any
verified component of the stack.}
\label{fig:arch}
\end{figure}

Three questions follow, and all three are answerable on public policy artefacts
without operational access logs.

\subsection*{Contributions}

\noindent\textbf{C1: Prose-derived deny sets are frequently unsound, for a
measurable structural reason, whoever writes them (Section~\ref{sec:c1}).}
We formalise soundness of $H$ against the grant relation, give an exhaustive
verification procedure, and apply it to all five public ABAC Lab policies, none
of which we authored. A candidate clause can be unsound for two unrelated
reasons, so we stratify by provenance: a clause that transcribes an
\emph{explicit} restriction (``only $X$ may \textsc{verb}'') isolates the scope
mismatch between local prose and a global grant disjunction, whereas a clause
that reads exclusivity into a plain permission is invalid before any policy is
consulted. The stratification is reproduced exactly by a fixed lexical rule, so
it does not depend on any verdict. In the explicit-restriction stratum the
author's clauses are unsound at $16/23$, and \emph{three blind language-model
elicitors} that never saw a rule or a verdict are unsound at $0.50$--$0.66$
(63/109 pooled); 15 restriction sentences fail for every elicitor that
formalises them. Attributing each of 137{,}454 counterexamples to the rules that
grant it confirms the mechanism directly: $97\%$ are granted by rules other than
the clause's own, $86\%$ by rules in other subsections. Shipped unchecked, the
four candidate sets would withhold $16$--$78\%$ of everything the policies
grant. We pair soundness with a \emph{non-vacuity} check, because 12 of the
independent clauses cover no request at all and pass soundness trivially.

\noindent\textbf{C2: A hybrid region map needs a kind of rule set that policy
miners do not produce, and the obvious way to build one is degenerate
(Section~\ref{sec:c2}).}
ABAC miners emit \emph{total} policies---permit rules under default deny, or
permit and deny rules under a conflict-resolution order---so a request that no
rule matches is decided, not abstained on, the informative region is empty, and
there is no routing decision to make. Hybrid routing presupposes a two-class,
\emph{abstaining} rule set. A Xu--Stoller-style greedy permit-cover arm makes
this concrete. We then show that the obvious construction---ranking a
mixed-effect candidate pool under a cap---is one-class in every cell for the two
commonest ranking keys across four policies and three cap sizes, because
minority-class support is $3.3$ to $72.7$ times below majority-class support.
We also show that a minimum confidence above $1-\varepsilon$ excludes correct
rules with a probability that sharpens exponentially in support, and test both
consequences under three noise processes that violate that result's
assumptions, reporting where the second fails to transfer.

\noindent\textbf{C3: Hybrid routing factorises exactly, and the routing verdict
cannot be read from one loss (Section~\ref{sec:c3}).}
For additive losses the hybrid-versus-learner effect equals covered mass times
the contrast between rule risk and learner risk on the covered region. The
consequence is a reporting requirement: the contrast is a different number for
each additive loss, because false-non-grant conditions on the true-grant
subpopulation and false-permit on the non-grant subpopulation, and a one-class
rule set is trivially optimal on whichever loss its single class cannot incur.
On E-Document, four of five candidate-selection rules give opposite routing
verdicts on the two losses, and only class-aware selection favours rules on
both, with contrasts $-0.033$ $[-0.049,-0.020]$ and $-0.115$ $[-0.127,-0.104]$.
That verdict holds under one of four learner configurations and only for a band
of noise levels, and it does not replicate on Workforce, where every selection
rule trades false non-grants for false permits. A study that reports one loss
can therefore reach either conclusion.

\noindent\textbf{Artefact.} An open analysis package regenerates every number in
this paper from the public policy files with one command and then checks each
value quoted in the manuscript against the regenerated results; it also
includes the multi-elicitor harness, the frozen elicitations and the exact
prompts used.

\subsection*{Scope and organisation}
The audit uses public policy artefacts and no operational deployment logs. C1
and C2 are established on public third-party policies; C3 is established
analytically and measured on two real policies. Section~\ref{sec:eval} measures
the composed system on the same two policies and then on an authored synthetic
policy, the only benchmark in our set with a five-action space and break-glass
semantics; the synthetic part demonstrates the measurement protocol and supports
no robustness claim. Section~\ref{sec:rw} reviews related work and
Section~\ref{sec:model} fixes the system and threat model.
Sections~\ref{sec:c1}--\ref{sec:c3} present C1--C3, Section~\ref{sec:eval} the
composed system, and Section~\ref{sec:threats} the threats to validity, including
what is assumed rather than established about system security
(Section~\ref{sec:securityscope}).

\section{Background and Related Work}
\label{sec:rw}

\subsection{Policy mining produces approximations}
ABAC rule mining recovers flat rule sets from logs
\cite{xu2015mining,cotrini2018sparse,iyer2018mining}, with extensions to
automatic extraction from access logs~\cite{karimi2022automatic}, hierarchical
clustering with relationship extraction~\cite{shang2024hierarchical}, learning
relationship-based policies by querying black-box
systems~\cite{iyer2022learning}, frameworks that combine data mining,
statistical and machine-learning techniques~\cite{abujabal2020polisma}, and
mining from massive, noisy logs~\cite{zhang2026domain}. Classifiers have also
been trained directly on access logs as policy surrogates and assessed for the
quality of the policies they imply~\cite{cappelletti2019quality,nobi2022deep}.
Guo \emph{et al.}~\cite{guo2025onthefly} state the premise directly: no mining
algorithm produces a completely accurate policy, and human revision is required.

Two properties of this literature matter for C2 and are easy to overlook.
First, the mined object is a \emph{total} policy. Most formalisms are
\emph{default-deny plus permit rules}, so a request no rule matches is denied,
not abstained on; miners that also learn negative
rules~\cite{iyer2018mining} still resolve every request through a default and a
conflict-resolution order. Second, selection is driven by coverage of the permit
tuples---Xu and Stoller's greedy phase takes the candidate covering the most
still-uncovered permits, and Cotrini \emph{et al.}'s Rhapsody uses a
sparse-log-specific quality function---not by ranking a mixed-effect candidate
pool under a size cap. The hybrid architecture needs an object these algorithms
do not emit, and Section~\ref{sec:c2} examines what happens when a pipeline
builds it the obvious way.

The underlying difficulty is not new in data mining. A single global minimum
support cannot simultaneously admit frequent and rare patterns; Liu \emph{et
al.}~\cite{liu1999msapriori} named this the \emph{rare item problem} and proposed
per-item minimum supports as the remedy, and the imbalance literature has studied
minority-class rule discovery since~\cite{weiss2004rarity}. C2's mechanism is
that problem instantiated in access control, and we make the connection explicit
rather than presenting it as a discovery. Two things are specific to this
setting. The minority class here is \emph{the permit class}, so the pathology
does not merely lose accuracy: it produces a rule set that denies everything,
which is silently plausible in a security context in a way that an all-negative
classifier is not elsewhere. And the degeneracy propagates into an architectural
decision: Section~\ref{sec:c3} shows the same choice moves the routing contrast
across an order of magnitude.

\subsection{Policy analysis, and what it is pointed at}
\label{sec:rw-analysis}
Deciding $\exists x:\ c(x)\wedge\Grant(P^{*}(x))$ over a finite policy is
property checking, and the access-control literature has done it for two
decades. Margrave performs change-impact analysis and verification queries over
XACML~\cite{fisler2005margrave}; Hughes and Bultan reduce XACML policy properties
to satisfiability~\cite{hughes2008xacml}; Turkmen \emph{et al.} formalise XACML
analysis in SMT~\cite{turkmen2017smt}. Industrially, Zelkova compiles AWS access
policies---\emph{including} their \texttt{Deny} statements---to SMT and answers
comparative permissiveness queries at
scale~\cite{backes2018zelkova}, Quacky \emph{counts} the requests a policy
permits via model counting to quantify
permissiveness~\cite{eiers2022quacky}, and Cedar supplies an expressive,
analyzable authorization language with a verified validator and a symbolic
compiler to SMT~\cite{cutler2024cedar}.

We claim no algorithmic novelty for the C1 check. Our certification
procedure is a decision problem these tools already solve, on smaller inputs, and
where the request space is not enumerable the check becomes a satisfiability
query that should be discharged by one of them. What is new is the \emph{object}
the check is pointed at and the resulting measurement. The verification
literature checks properties an analyst asserts about a policy; C1 checks the
\emph{hard-deny clause set that a trusted enforcement layer will execute},
elicited from the policy's prose, and reports how often such clauses over-claim
against the very relation they are meant to constrain---stratified by how the
clause was derived, which is what makes the rate interpretable. To
our knowledge no prior work reports a soundness rate for prose-derived deny sets
on public policy artefacts, nor identifies vacuity as its opposite failure mode.

\subsection{Deny rules, prose translation, and zero trust}
Iyer \emph{et al.}~\cite{iyer2022expressive} show that negated conditions and
negative authorizations are not interchangeable and that the choice constrains
what a policy model can express. Their question is expressiveness; ours is
whether the clauses someone actually writes into a negative-authorization set
are sound against the grant relation they are supposed to constrain.

The step C1 audits---a human reading policy prose and writing an enforceable
predicate---is increasingly automated. Prior work extracts ABAC attributes
and models from natural-language policies~\cite{alohaly2019nlp,abdelgawad2023synthesizing}
and ABAC rules from access logs~\cite{karimi2022automatic}, and large language
models are now used to extract enforceable ABAC policies from
text~\cite{yang2025llmabac}, to synthesise access-control
policies~\cite{vatsa2025synthesizing}, and to mine ABAC
policies~\cite{bui2026llmmining}. These works evaluate generated policies
against reference policies or by testing; none certifies a generated
\emph{deny} set exhaustively against the grant relation it is meant to
constrain. Section~\ref{sec:independent} does that for three model elicitors,
under the same freeze-then-certify protocol as the human elicitor.

Zero-trust architectures approach complete mediation from the deployment side.
Al-hammuri \emph{et al.}~\cite{alhammuri2023zerotrust} score trust from
contextual features across authentication, encryption, logging and
authorization. The premise we share is that no request inherits authority from a
prior decision. The difference is where the guarantee is located: a trust score
is a graded, learned judgement, and by our own argument such a component belongs
\emph{outside} the trusted computing base. Buccafurri \emph{et
al.}~\cite{buccafurri2024enforcing} show that interacting authentication systems
implicitly extend each other's policies, so a system adopting strong policies can
be unexpectedly weak; we take this seriously in stating our composition result as
regime-indexed rather than unconditional.

\subsection{Residual over-granting, selective prediction, benchmarks}
Closest in spirit to our residual-surface analysis is
OPMonitor~\cite{wang2025opmonitor}, which observes that policy verification is
bounded by the comprehensiveness of predefined verification goals, so
over-granted permissions outside those goals remain undetected, and monitors
residual over-granting against a log-inferred baseline. Their unverified
permissions correspond structurally to our residual surface; the treatments are
complementary. Madkaikar \emph{et al.}~\cite{madkaikar2025queuing} analyse
dynamic ABAC performance in a queuing-theoretic framework; we do not measure
enforcement latency, and refer to that line of work for it.

Selective classification with a reject option~\cite{chow1970,geifman2017selective}
trades coverage for accuracy; conformal prediction supplies distribution-free
finite-sample guarantees under
exchangeability~\cite{vovk2005alrw,shafer2008tutorial}, with label-conditional
variants~\cite{vovk2012conditional,angelopoulos2023gentle} and conformal risk
control for bounded losses~\cite{angelopoulos2024crc}. Our use of these tools is
instrumental, and the marginal-versus-conditional distinction we rely on is
theirs, not ours. Label-noise learning~\cite{natarajan2013noisy,frenay2014survey}
supplies the setting for Proposition~\ref{prop:ceiling}, whose content is the
consequence for a \emph{confidence threshold} rather than for a classifier.

Real access-control policies and logs are rarely public, and ML-based access
control specifically suffers a shortage of public
datasets~\cite{aboukadri2024survey}. The Xu--Stoller case
studies~\cite{xu2015mining} and the ABAC Lab consolidation~\cite{bui2025abaclab},
which adds instances derived from real-world case studies including the KU Leuven
e-document study~\cite{decat2014edocument}, are the community's response. C1 and
C2 are answers to the data problem rather than victims of it: both are
established on public policy artefacts, and neither requires a deployment log.
Break-glass access is a recognised clinical
requirement~\cite{brucker2009breakglass,rostad2006audit}, and supplies the
semantics of the emergency branch of Section~\ref{sec:eval}.

\section{System Model and Threat Model}
\label{sec:model}

\subsection{Requests, actions, policies}
A request is $x\in\mathcal{X}$ over subject, object, operation and environment
attributes. The action space is $\mathcal{A}=\{\textsc{deny},\textsc{permit},
\textsc{monitor},\textsc{temporary},\textsc{escalate}\}$, with
\begin{equation}
\Grant(a)\iff a\in\{\textsc{permit},\textsc{monitor},\textsc{temporary}\},
\label{eq:grant}
\end{equation}
so \textsc{deny} and \textsc{escalate} are non-granting. On the public policies,
which are binary, $\mathcal{A}$ degenerates to $\{\textsc{deny},
\textsc{permit}\}$. $P^{*}$ is the latent policy and the sole label source. $R$
is a mined rule set inducing $\mathrm{region}(x)\in\{\text{covered},\text{gap},
\text{conflict}\}$: \emph{covered} when the matching rules of $R$ agree on one
effect, \emph{gap} when none matches, \emph{conflict} when two or more disagree.
$L$ is the untrusted learner, and $\alpha(x)$ the possibly corrupted view of $x$
that the enforcement layer receives.

\begin{definition}[Hard-deny set]
$H=\{x:\hardDeny(x)\neq\bot\}$, where $\hardDeny$ is a disjunction of clauses
each of which independently triggers hard denial, with codomain
$\{\textsc{deny},\textsc{escalate}\}$.
\end{definition}

\begin{definition}[Soundness of $H$]
\label{def:sound}
$H$ is \emph{sound with respect to $P^{*}$} iff $\forall x\in H,\ P^{*}(x)\in
\{\textsc{deny},\textsc{escalate}\}$, equivalently $\nexists x:\ x\in H\wedge
\Grant(P^{*}(x))$.
\end{definition}

Write $\Hreal(x)=\mathbf{1}[x\in H]$ for membership of the true request and
$\Habs(\alpha(x))=\mathbf{1}[\hardDeny(\alpha(x))\neq\bot]$ for the hard-deny
predicate evaluated on the view the enforcement layer actually receives. All
probabilities are over a request $X$ drawn from the deployment distribution.

\begin{definition}[Directional abstraction errors]
$\epsm=\Pr[\Habs(\alpha(X))=0\mid\Hreal(X)=1]$ and
$\epsp=\Pr[\Habs(\alpha(X))=1\mid\Hreal(X)=0]$.
\end{definition}
Only $\epsm$ can produce an unsafe grant; $\epsp$ forces denials and appears on
the availability side. The enforcement layer is
\begin{equation}
\sigma(x,a)=\begin{cases}
\hardDeny(\alpha(x)) & \Habs(\alpha(x))=1,\\
a & \text{otherwise,}
\end{cases}
\label{eq:sigma}
\end{equation}
with hard deny evaluated before any other branch. Let $\Sg=\{x:\neg\Grant(P^{*}(x))\}$
be the set of requests on which any grant is unsafe. Residual quantities are
reported on the \emph{residual surface} $\Sres=\Sg\setminus H$, the unsafe
requests that $H$ does not name.

\subsection{Assumptions and trust boundaries}
\textbf{(A1)} $P^{*}$ defines truth (evaluation only; not used in any enforcement
result). \textbf{(A2)} Pointwise clean-input preservation, $\forall x,\
\Habs(\alpha(x))=\Hreal(x)$; this is strictly stronger than $\epsm=0$ and is
dropped in Proposition~\ref{prop:directional}. \textbf{(A3)} Shield code is
correct and executes on every request.

The adversary may control the learner entirely, including its training data and
outputs, and may manipulate request attributes that are not attested. Two trust
boundaries must be distinguished. The \emph{deterministic safety TCB}---shield
code, hard-deny clause set, attested attribute store, actuator path---is what
Lemma~\ref{lem:conform} and Proposition~\ref{prop:directional} depend on. Any
statistical gate placed upstream of $\sigma$ is \emph{not} in it: because
$\sigma$ runs last, a compromised gate cannot produce a grant on $H$.
Table~\ref{tab:tcb} states the status of each component and, in the column that
matters when deciding what to verify, what a compromise of it buys the adversary.

\begin{table}[t]
\centering
\caption{Components and trust status. The hard-deny clause set is in the safety
TCB, and no amount of verification of the shield or of an authorization language
detects an error in it. Section~\ref{sec:c1} audits exactly this row.}
\label{tab:tcb}
\footnotesize
\setlength{\tabcolsep}{3pt}\begin{tabular}{@{}>{\raggedright\arraybackslash}p{3.0cm}>{\raggedright\arraybackslash}p{2.6cm}>{\raggedright\arraybackslash}p{3.6cm}>{\raggedright\arraybackslash}p{5.0cm}@{}}
\toprule
component & status & protects & failure consequence \\
\midrule
learner $L$ & untrusted & --- & none on $H$; bounded by $\sigma$ \\
shield $\sigma$ & safety TCB & conformance on $H$ & unbounded grants \\
\textbf{hard-deny clauses} & \textbf{safety TCB} & \textbf{soundness of $H$} &
\textbf{wrong denials (\S\ref{sec:c1})} \\
attested store & safety TCB & attribute integrity & $\epsm>0$ \\
actuator path & safety TCB & complete mediation & shield bypass \\
statistical gate & guarantee only & marginal FP rate & bound void; $H$ unaffected \\
audit log & recording only & completeness & undetected break-glass \\
\bottomrule
\end{tabular}
\end{table}

\subsection{What the trusted layer does and does not bound}
\label{sec:theory}

\begin{lemma}[Conformance and residual invariance]
\label{lem:conform}
Under (A2) and (A3): (i) for all $x\in H$ and all $a\in\mathcal{A}$,
$\neg\Grant(\sigma(x,a))$; (ii) for all $x\in\Sres$ and all $a$,
$\sigma(x,a)=a$, so the unsafe-grant rate on $\Sres$ is identical with and
without $\sigma$.
\end{lemma}
\begin{proof}
(i) By (A2), $\Habs(\alpha(x))=1\iff x\in H$, so the first branch of
\eqref{eq:sigma} fires and returns $\hardDeny(\alpha(x))\in\{\textsc{deny},
\textsc{escalate}\}$, non-granting by \eqref{eq:grant}; by (A3) the branch is
reached on every request. The conclusion needs only that the returned action be
non-granting, not that $\hardDeny(\alpha(x))=\hardDeny(x)$. (ii) By (A2) at
$x\in\Sres$, $\Habs(\alpha(x))=\Hreal(x)=0$---this is the step requiring the
absence of hard-deny false positives---so no other branch applies.
\end{proof}

Lemma~\ref{lem:conform} records that the shield's definition has its intended
effect. It is not a novel security theorem and we do not present it as one; we
state it because it fixes what the rest of the paper is about. Part~(i) says the
shield is exactly as good as $H$, and says nothing about whether $H$ is
\emph{correct}---which is C1. Part~(ii) says the shield cannot help on $\Sres$, so
reducing residual exposure requires naming more of the surface, which C1 shows is
hazardous.

\begin{proposition}[Directional conditional guarantee]
\label{prop:directional}
Drop (A2) entirely. Under (A3) alone,
$\Pr[\Grant(\sigma(X,L(X)))\mid\Hreal(X)=1]\leq\epsm$; when $\epsm>0$, equality
holds iff $\Pr[\Grant(L(X))\mid\Hreal(X)=1,\Habs(\alpha(X))=0]=1$.
\end{proposition}
\begin{proof}
Partition $\{\Hreal(X)=1\}$ on whether the abstraction preserves the hard deny.
On $\{\Habs(\alpha(X))=1\}$, of conditional probability $1-\epsm$, $\sigma$
returns a non-granting action, so the grant probability is $0$. On the
complement, of conditional probability $\epsm$, $\sigma$ returns $L(X)$, which
grants with conditional probability $p\leq1$. The unsafe-grant rate is therefore
$\epsm p\leq\epsm$, with equality iff $p=1$ whenever $\epsm>0$.
\end{proof}

Hard-deny false positives ($\epsp$) cannot produce a grant; they force denials
and belong in the availability column. A symmetric quantity $\Pr[\Habs\neq
\Hreal]$ mixes the two and is neither an upper bound on nor a proxy for the
unsafe-grant rate. This is a \emph{reporting prescription} entailed by the
partition, not a discovery, and Section~\ref{sec:eval} reports the three
quantities side by side so a reader can see how far the symmetric one would
mislead.

\section{C1: Deny-Set Soundness}
\label{sec:c1}

\subsection{The mechanism}
Defining $H$ is the one design step no part of the pipeline automates. A human
reads the policy's prose and writes clauses. The failure mode we target is
specific: \emph{policy descriptions state local restrictions, while the grant
relation is a global disjunction.} A sentence of the form ``only members of $X$
may \textsc{verb}'' is almost always scoped to the subsection it appears in,
whereas the induced grant relation is the union over \emph{all} rules, several of
which typically cut across that subsection. A clause read off such a sentence
over-claims, and admits into $H$ requests the policy grants.

The consequence is asymmetric and easy to miss. An unsound clause does not
produce a visible error: it produces a \emph{denial}, issued by the trusted
layer, with the shield's authority, on a request that should have been granted.
It raises the false-non-grant rate, appears only in the availability column, and
is indistinguishable from a conservative learner unless the deny set is checked
directly. (It does not raise $\epsp$, which is measured against $H$ itself.)

\subsection{The check}
A candidate clause $c$ is admitted to $H$ iff
\begin{equation}
\nexists\, x\in\mathcal{X}:\ c(x)\wedge \Grant(P^{*}(x)),
\label{eq:soundcheck}
\end{equation}
verified by exhaustive enumeration over the full grant relation, not by
sampling. On a policy with a finite enumerable request space the check is
decidable and cheap. On a 2014 laptop CPU (Intel Core i7-4610M, Python~3,
NumPy) the full 600{,}000-tuple E-Document relation is parsed and built in under
$0.5$\,s, and all 55 certificates of Section~\ref{sec:c1results} (49 candidates
and six vacuity-pass clauses, across five policies) are computed in about
$0.1$\,s.

Sampling is not an acceptable substitute. The counterexample sets below range
over four orders of magnitude, and a clause with 400 counterexamples in 600{,}000
tuples has a per-sample detection probability below $10^{-3}$: a 1{,}000-request
sample misses it with probability ${\approx}0.51$. The check is a \emph{gate},
not a diagnostic---a clause that fails is rejected, and the analyst either scopes
it more tightly or discards it. Where the request space is not enumerable the
check becomes the satisfiability query $c(x)\wedge\Grant(P^{*}(x))$ and inherits
the completeness of the solver; the symbolic compilers
of~\cite{backes2018zelkova,cutler2024cedar} are the natural integration point,
and we flag that rather than claiming it.

\begin{figure}[t]
\centering
\includegraphics[width=0.7\linewidth]{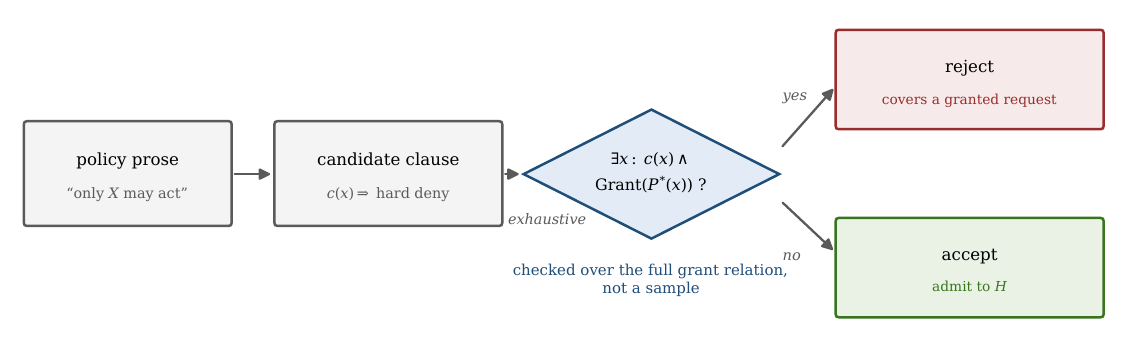}
\caption{Hard-deny soundness as a first-class check. A clause read off the
policy's prose is admitted to $H$ only if it covers \emph{no} request the policy
grants, verified exhaustively against the full grant relation rather than a
sample. Because policy descriptions state \emph{local} restrictions while the
grant relation is a \emph{global} disjunction, prose-derived restriction clauses
over-claim by default (Table~\ref{tab:sound}). An admitted clause must
additionally be checked for non-vacuity (Section~\ref{sec:vacuity}), which the
same enumeration pass supplies at no extra cost.}
\label{fig:hsound}
\end{figure}

\subsection{Audit protocol}
\label{sec:auditprotocol}
So that the result is a property of the failure mode rather than of one analyst
or one policy, we fix the following protocol and apply it identically everywhere.

\begin{enumerate}
\item \textbf{Candidate elicitation.} The elicitor reads only the policy's
\emph{prose} description---the case-study narrative and rule commentary, not the
formal rule set---and writes every sentence that reads as a categorical
prohibition as a candidate clause. Reading the formal rules first would leak the
grant relation into the clause and destroy the measurement. The packet is
produced deterministically by stripping ABAC Lab's fenced formal-rule blocks from
the published README, is leak-checked (no \texttt{subCond}, \texttt{resCond},
$\in$, $\ni$ or $\supseteq$ survives), and carries a SHA-256.
\item \textbf{Provenance tagging.} Each candidate records \emph{how it was
derived from its source sentence} (Section~\ref{sec:provenance}). This step is
new here and is what makes the results below interpretable.
\item \textbf{Freeze.} Candidates are recorded with a per-clause content
fingerprint before certification; for further elicitors the released harness
writes a UTC-timestamped freeze and refuses to overwrite an existing one.
\item \textbf{Certification.} Each frozen candidate is checked by
\eqref{eq:soundcheck} over the full enumerated request space, recording the
counterexample count, the coverage, and for rejected clauses one witness.
\item \textbf{Repair.} Each rejected clause is offered one scoping repair---the
narrowest restriction the prose actually supports---and re-certified. A repair is
recorded \emph{as a repair} and is excluded from the candidate denominator, since
counting it separately would double-count one sentence.
\end{enumerate}

\subsection{Provenance: why a pooled rejection rate means nothing}
\label{sec:provenance}
A candidate clause can be unsound for two entirely unrelated reasons, and a
pooled rate credits one with the other. Consider two sentences from our corpus:

\begin{quote}\footnotesize
(a) \emph{``Only members of the sales department can send, view, or search
invoices.''} (E-Document, Rule 7)\\
(b) \emph{``The author of an item can read it.''} (Healthcare, Rule 5)
\end{quote}

Reading (a) as ``deny \textsc{send}/\textsc{view}/\textsc{search} on invoices
when the subject is not in Sales'' is a faithful transcription: the sentence
asserts exclusivity, and if the resulting clause covers a granted request, the
reason is that the sentence's ``only'' is scoped to its subsection while the
grant relation is not. That is the mechanism of Section~\ref{sec:c1}. Reading (b)
as ``deny \textsc{read} when the subject is not the author'' is not a
transcription at all: the sentence grants a permission and says nothing about
exclusivity. Such a clause is invalid before any policy is consulted, and its
rejection measures the elicitor, not the artefact.

We therefore tag every candidate with its derivation and report the strata apart:

\begin{itemize}
\item \textbf{explicit-restriction} --- the sentence itself states a restriction
(``only $X$ may \textsc{verb}'', ``$X$ may only \textsc{verb}'', ``\dots except
$Y$''). \emph{This is the mechanism stratum.}
\item \textbf{inferred-exclusivity} --- ``$X$ can \textsc{verb}'' read as ``only
$X$ can \textsc{verb}''.
\item \textbf{inferred-necessity} --- a sufficient condition for a grant read as
a necessary one.
\item \textbf{silence} --- the clause rests on the \emph{absence} of a stated
permission; an argument from silence.
\item \textbf{repair} --- a re-scoped variant of another candidate; excluded from
the denominator.
\end{itemize}

The tags are assigned from the source sentence alone. Because the elicitor
also saw the verdicts of this audit, we check that the stratum assignment that
matters---mechanism or not---does not depend on them: a fixed lexical rule that
marks a sentence explicit iff it contains \emph{only}, \emph{except} or
\emph{unless} reproduces the hand assignment on all 46 candidates that have a
source sentence (23/23 explicit, 23/23 non-explicit; the two \emph{silence}
candidates have no source sentence, and the repair is excluded). The mechanism
stratum is therefore a deterministic function of the policy text. We regard this
stratification, rather than any single rate, as C1's methodological
contribution: a rejection rate for a prose-derived deny set is uninterpretable
without it.

\subsection{Results}
\label{sec:c1results}
Table~\ref{tab:sound} reports the audit. All checks are exhaustive over the
enumerated relation, so the certificates are exact and admit no false negatives
over that relation. The parser is a faithful reimplementation of ABAC Lab's own
evaluator, validated at run time by reproducing each policy's published
statistics to the unit---E-Document's 600{,}000-tuple space with
$\#\mathrm{perm}=32{,}961$, and Healthcare's 1{,}008 requests with 43 grants.

\begin{table}[t]
\centering
\caption{Hard-deny soundness audit across all five public ABAC Lab policies,
stratified by clause provenance.}
\vspace{-2pt}
{\footnotesize\raggedright ``expl.'' is the explicit-restriction (mechanism)
stratum of Section~\ref{sec:provenance}; ``other'' pools the three
elicitation-error strata. The three Xu--Stoller policies have an \emph{empty}
mechanism stratum: their prose contains no restriction statement, so the
mechanism cannot be exercised on them and every rejection there is an elicitation
error. One repair (E-Document) is excluded from the denominator.\par}
\vspace{2pt}
\label{tab:sound}
\footnotesize
\begin{tabular}{@{}l@{\hskip 4pt}r@{\hskip 5pt}c@{\hskip 5pt}c@{\hskip 5pt}l@{}}
\toprule
policy & enum.\ tuples & expl. & other & expl.\ c/ex \\
\midrule
E-Document~\cite{decat2014edocument} & 600{,}000 & 12/14 & 1/2 & 495--6{,}489 \\
Workforce~\cite{bui2025abaclab}      & 794{,}250 &  4/9  & 0/1 & 20--7{,}665 \\
University~\cite{xu2015mining}       &   6{,}732 &  ---  & 3/8 & --- \\
Project Mgmt~\cite{xu2015mining}     &   3{,}040 &  ---  & 2/6 & --- \\
Healthcare~\cite{xu2015mining}       &   1{,}008 &  ---  & 7/8 & --- \\
\midrule
\textbf{total} & & \textbf{16/23} & \textbf{13/25} & 20--7{,}665 \\
\bottomrule
\end{tabular}
\end{table}

\begin{table}[t]
\centering
\caption{Rejection rate by provenance, with exact Clopper--Pearson intervals
(reported rather than a normal approximation because several strata have
single-digit denominators). The strata are not pooled: only the first
isolates the scope mismatch, and the pooled rate of $0.60$ is reported solely to
show what a study omitting this stratification would have published.}
\label{tab:prov}
\footnotesize
\begin{tabular}{@{}lrrl@{}}
\toprule
provenance & $n$ & unsound & rate [95\% CI] \\
\midrule
\textbf{explicit-restriction} (mechanism) & 23 & 16 & \textbf{0.696} $[0.471,0.868]$ \\
inferred-exclusivity   & 12 &  7 & 0.583 $[0.277,0.848]$ \\
inferred-necessity     & 11 &  5 & 0.455 $[0.168,0.766]$ \\
silence                &  2 &  1 & 0.500 $[0.013,0.987]$ \\
\midrule
elicitation-error strata (pooled) & 25 & 13 & 0.520 $[0.313,0.722]$ \\
\emph{all strata pooled (uninterpretable)} & 48 & 29 & \emph{0.604} $[0.453,0.742]$ \\
\bottomrule
\end{tabular}
\end{table}

Sixteen of twenty-three faithfully transcribed restriction clauses cover
requests the policy grants (Table~\ref{tab:prov}), a rate of $0.70$ with an
exact interval of $[0.47,0.87]$, and counterexample counts from 20 to 7{,}665
over relations of up to 794{,}250 tuples. Section~\ref{sec:independent} tests
whether this rate is a property of the elicitor.

The mechanism stratum is empty on three of five policies, and this is itself
informative. The Xu--Stoller instances~\cite{xu2015mining} document their
policies as lists of granting rules and never write a restriction; the two
multi-tenant case-study-derived instances~\cite{decat2014edocument,bui2025abaclab}
write restrictions constantly, because a tenant subsection is exactly the context
in which an author reaches for ``only''. On this corpus the failure mode
therefore tracks how multi-tenant policies are documented rather than ABAC as
such, and suggests where a practitioner should look first: \emph{cross-sectional universals in prose
that the rule set does not enforce universally}.

\subsection{Independent, blind elicitation}
\label{sec:independent}
The rates above come from one elicitor, an author who also saw the verdicts. To
test whether they are a property of that elicitor, we ran three further
elicitors that saw nothing but what the protocol allows. Each, L1--L3, was a
fresh language-model agent---Claude Opus~5.5, Sonnet~5.5 and Haiku~5.5, one
model each---with read access to exactly five files: the elicitation prompts
that the released harness generates for the five policies (prose packet and
attribute schema, with a syntax example taken from a fictional library policy).
No agent had access to the formal rules, the certificates, the code, or any
other elicitor's clauses, and each confirmed reading only those files. The
candidates were parsed by the harness, frozen with per-clause fingerprints and
a timestamp, and only then certified. Every one of the 113 quoted source
sentences was found in the packet the elicitor was shown, and each candidate was
assigned to the mechanism stratum by the lexical rule of
Section~\ref{sec:provenance}, so no hand tagging enters this comparison.

\begin{table}[t]
\centering
\caption{Author and independent elicitors compared. Mechanism stratum
(explicit-restriction) with exact Clopper--Pearson intervals; ``vac.'' counts
clauses that cover no request; the last two columns give the fraction of all
granted requests that the elicitor's unchecked candidate set would deny if
shipped as $H$.}
\label{tab:indep}
\footnotesize
\setlength{\tabcolsep}{2.6pt}
\begin{tabular}{@{}llrlrcc@{}}
\toprule
 & & \multicolumn{2}{c}{mechanism stratum} & & \multicolumn{2}{c}{grants denied}\\
\cmidrule(lr){3-4}\cmidrule(lr){6-7}
elicitor & model & unsound & rate [95\% CI] & vac. & E-Doc. & Workf.\\
\midrule
R1 & author        & 16/23 & 0.696 $[0.471,0.868]$ & 0 & 0.516 & 0.548 \\
L1 & Opus 5.5      & 18/36 & 0.500 $[0.329,0.671]$ & 4 & 0.159 & 0.546 \\
L2 & Sonnet 5.5    & 22/38 & 0.579 $[0.408,0.737]$ & 5 & 0.265 & 0.593 \\
L3 & Haiku 5.5     & 23/35 & 0.657 $[0.478,0.809]$ & 3 & 0.702 & 0.780 \\
\midrule
\multicolumn{2}{@{}l}{L1--L3 pooled} & 63/109 & 0.578 $[0.480,0.672]$ & 12 & & \\
\bottomrule
\end{tabular}
\end{table}

\emph{The rate is not a property of the author.} All three independent
elicitors write unsound restriction clauses at rates comparable to the author's
(Table~\ref{tab:indep}): $0.50$, $0.58$ and $0.66$ against $0.70$, with
overlapping intervals, and $63$ of $109$ pooled. The largest model, which
deliberately scoped many of its clauses to the tenant named in the section
heading, still failed on half. The hypothesis left open by a single elicitation,
that independent elicitors write sound restriction clauses, is rejected on this
corpus.

\emph{The same sentences fail.} The README contains 31 rule sentences that the
lexical rule marks as explicit restrictions. Of these, 27 were formalised by at
least two of the four elicitors; on 20 every elicitor's clause received the same
verdict, and on 15 every elicitor's clause was unsound. Nine E-Document
restrictions (Rules 1, 7--11, 15, 17 and 19) are unsound for all four
elicitors. Agreement on which sentences to formalise is only moderate (Fleiss'
$\kappa=0.42$ over the 31 sentences), so elicitors differ more in what they
notice than in whether what they write survives certification.

\subsection{Where the counterexamples come from}
\label{sec:attribution}
The mechanism of Section~\ref{sec:c1} makes a testable prediction: an unsound
restriction clause is contradicted not by its own rule but by other rules, from
other parts of the policy. Every counterexample is granted by at least one rule,
and the README places each rule in a section (a tenant, the \texttt{\#\#\#}
heading) and a subsection (the \texttt{\#\#\#\#} heading beneath it), so the
prediction can be checked counterexample by counterexample. Over the 77 unsound
mechanism-stratum clauses of the four elicitors whose source rule can be located
(of 79; 137{,}454 counterexamples),
only $3.4\%$ of counterexamples are granted by the clause's own source rule
($0\%$ for R1), $86\%$ are granted only by rules in \emph{other subsections}, and
$62\%$ only by rules of \emph{other tenants}.

The unsoundness is therefore almost entirely the disjunction with other rules,
as the mechanism predicts, and most of it crosses a subsection rather than a
tenant boundary. This has a practical consequence. Scoping a clause to its
tenant---re-inserting the most visible section boundary, by adding a guard such
as \texttt{resource.tenant = largeBank}---repairs only 11 of the 79 unsound
clauses, because a sentence's ``only'' is scoped to a subsection for which the
object model has no attribute.

Workforce shows both cases. It has two sentences about the \emph{same resource
type}, in different tenant sections:
\begin{quote}\footnotesize
Rule 15 (PowerProtection): \emph{Only members of the Provisioning group can
receive stock refill notifications.}\\
Rule 19 (TelCo): \emph{Only members of Technician Support can receive stock
refill notifications.}
\end{quote}
Read globally, as R1 read them, they are structurally identical but receive
opposite verdicts. ``Only Technician Support'' covers 6{,}840 tuples and admits
no granted request. ``Only Provisioning'' covers 6{,}400 tuples and has 20
counterexamples, witnessed by \texttt{(slmgr001, stockrefill001, receive)}---a
sales manager in the TelCo tenant and the \texttt{techSupport} group, receiving
a refill request granted by Rule~19. All three independent elicitors scoped both
sentences to their tenant, and all three of their Rule-15 clauses are sound.
This is the case tenant scoping fixes. E-Document shows the case it does not:
Rule~7 (``Only members of the sales department can send, view, or search
invoices'') is contradicted by Rule~12, which lets auditors in the \emph{same}
Large Bank section read invoices, and every elicitor's Rule-7 clause is unsound,
including the two that were scoped to that tenant. In both cases the certificate is
computed in well under a second, and it is the only step in the pipeline that
distinguishes the sound reading from the unsound one.

\subsection{A second failure mode: vacuity}
\label{sec:vacuity}
A clause that covers no request passes the soundness check trivially: it
certainly covers no granted request. An analyst would ship it believing $H$
constrains something it does not touch. The failure is silent in exactly the way
an unsound clause is silent, in the opposite direction: an unsound clause denies
too much, a vacuous one denies nothing while appearing to.

It occurs in practice. Twelve of the 113 independent clauses are vacuous
(Table~\ref{tab:indep}). In each case the conjunction of guards selects a
population the object model does not instantiate. E-Document's Rule~25 is the
clearest: it restricts ``Registered Private Receivers'', but all 16 private
receivers in the object model are unregistered, so every faithful formalisation
of the sentence---all three independent elicitors wrote one---covers nothing.
A second, deliberately constructed elicitor shows the vocabulary route to the
same failure: formalising each sentence with the informal names the prose uses
(``stock refill notifications'' for the resource type
\texttt{stockRefillRequest}, ``create'' where the object model has only
\texttt{createOneTimeWorkOrder} and \texttt{createRecurrentWorkOrder}) produced
six clauses, all vacuous and all certified sound. We therefore pair the
soundness certificate with a \emph{non-vacuity} check---the clause must cover at
least one request---and report coverage for every admitted clause. Both come
from the same enumeration pass at no extra cost.

\subsection{How much of a verdict is the scoping choice?}
\label{sec:scopesens}
Within the mechanism stratum the remaining degree of freedom is \emph{scoping}:
given the same restriction sentence, elicitors may draw the clause's boundary
more or less tightly. For every clause R1 admitted as sound we enumerate its
one-conjunct \emph{broadenings}---the clause obtained by dropping a single guard,
i.e.\ what the same sentence yields if read slightly less narrowly---and certify
each exhaustively. In the mechanism stratum, 14 of 21 such variants are unsound
($0.67$, $[0.43,0.85]$), with counterexample counts from 16 to 16{,}202; across
all strata, including variants of the repaired clause, it is 40 of 59. Two
thirds of admitted restriction clauses therefore sit directly on the soundness
boundary: they are sound because of a guard the elicitor happened to include,
and one step outwards crosses into covering granted requests. This is the
single-elicitor counterpart of the between-elicitor variation in
Section~\ref{sec:independent}.

\subsection{The cost of skipping the check}
\label{sec:cost}
An unsound clause produces a denial with the trusted layer's authority, and
nothing in the enforcement stack reports it. The last two columns of
Table~\ref{tab:indep} measure what that costs. Shipped unchecked as $H$, the
candidate sets of the four elicitors would deny between $16\%$ and $70\%$ of all
requests E-Document grants, and between $55\%$ and $78\%$ of those Workforce
grants. The certified subsets of the same candidate sets deny none, by
construction, while still naming up to $3\%$ of E-Document's and between $15\%$
and $50\%$ of Workforce's non-granted requests. Section~\ref{sec:eval-real}
shows the same contrast inside a composed system with a trained learner.

\subsection{What C1 establishes}
\begin{enumerate}
\item \textbf{The mechanism is structural, and measurable.} Where prose states a
restriction scoped to a section while the grant relation is a disjunction over
all rules, a clause read from that sentence over-claims. Counterexample
attribution confirms it: $97\%$ of counterexamples are granted by other rules,
$86\%$ by rules in other subsections.
\item \textbf{The rate is not one elicitor's.} In the stratum where the mechanism
is isolated, four elicitors wrote unsound clauses at rates of $0.50$--$0.70$,
and 15 sentences failed for every elicitor that formalised them. These are
counts for two public artefacts and four elicitors, not a prevalence estimate.
\item \textbf{A pooled rate is not a measurement.} Had we reported R1's pooled
$29/48$ we would have credited the mechanism with 13 rejections caused by
invalid inferences. Stratification is reproducible by a lexical rule, and we
supply it.
\item \textbf{Soundness alone is insufficient}: it must be paired with a
coverage figure, since vacuity occurs in independent elicitations too
(Section~\ref{sec:vacuity}).
\item \textbf{Skipping the check is costly and silent}: an unchecked prose-derived
deny set can withhold most of what a policy grants (Section~\ref{sec:cost}).
\end{enumerate}
We recommend that any paper reporting a shielded or deny-first architecture over
a named set report the soundness certificate for that set, together with a
coverage figure, the size of the enumerated relation, and the elicitation
protocol and provenance stratification that produced the candidates.

\section{C2: Constructing a Region Map}
\label{sec:c2}

The region map---which requests are covered by $R$ and which are handed to the
learner---is produced by a mining pipeline with hyperparameters. This section
makes two claims. The first is that the object the hybrid architecture needs is
not the object the mining literature produces, so there is no established way to
build it. The second is that two hyperparameters of the obvious construction have
failure modes that are arithmetic and combinatorial rather than data-limited, and
therefore cannot be fixed by collecting more logs.

\subsection{A hybrid region map is not a mined ABAC policy}
\label{sec:reference}
A mined ABAC policy is \emph{total} (Section~\ref{sec:rw}). In the commonest
formalism---default deny plus permit rules---the candidate pool is also
one-class, so there is no minority class to crowd out; a request no rule matches
is denied, because that is what the policy language means. Miners that also
learn negative rules~\cite{iyer2018mining} are two-class but remain total,
deciding unmatched or conflicting requests by a default and a
conflict-resolution order. Under either semantics the informative region is
empty, $m=0$, Proposition~\ref{prop:decomp} gives an effect of exactly zero, and
there is no routing decision to make.

The hybrid architecture therefore presupposes something else: a \emph{two-class,
abstaining} rule set, one that permits on some requests, denies on others, and
declines to answer on the rest. We are not aware of an ABAC miner that emits
such an object. To make this concrete we implement a reference arm---a
Xu--Stoller-style greedy cover over the permit tuples, permit-effect candidates
only, selection by marginal permit coverage rather than by any ranking under a
cap---and evaluate the \emph{same} candidate pool under three readings
(Table~\ref{tab:reference}).

\begin{table}[t]
\centering
\caption{The published-miner reference arm, read three ways.}
\vspace{-2pt}
{\footnotesize\raggedright Under default-deny the reference miner produces a
usable policy and there is no routing decision ($m=0$). Read as abstaining, the
same rules cover 2--5\% of the space and are all-\emph{permit}. Cap-and-rank over
a mixed-effect pool is all-\emph{deny}. Three natural constructions, three
different failures. Workforce is mined at $L=2$ over a 200{,}000-row sample
because its 25 attributes make length-3 enumeration intractable here; at that
depth the greedy cover finds one rule, and its row is a limitation of the
hypothesis language, not of the reference algorithm.\par}
\vspace{2pt}
\label{tab:reference}
\footnotesize
\begin{tabular}{@{}lrrrrr@{}}
\toprule
 & & \multicolumn{2}{c}{default-deny (total)} & \multicolumn{2}{c}{open-world (abstain)}\\
\cmidrule(lr){3-4}\cmidrule(lr){5-6}
policy & rules & permit-F1 & $m$ & covered & deny-F1 \\
\midrule
E-Document         & 35 & 0.934 & 0.000 & 0.048 & 0.000 \\
University         & 10 & 0.930 & 0.000 & 0.022 & 0.000 \\
Project Mgmt       &  8 & 0.865 & 0.000 & 0.025 & 0.000 \\
Workforce$^{\dagger}$ &  1 & 0.006 & 0.000 & 0.000 & 0.000 \\
\bottomrule
\multicolumn{6}{@{}l@{}}{\footnotesize $^{\dagger}$ $L=2$ over a 200{,}000-row sample; see caption.}
\end{tabular}
\end{table}

The reference miner produces a good policy on three of four policies
(permit-F1 $0.87$--$0.93$) and no region map at all. Read as abstaining, it
covers under 5\% of the request space and its covered region is one-class in the
permit direction. Cap-and-rank over a mixed-effect pool, examined next, is
one-class in the deny direction. The contribution of C2 is therefore not that
published miners are broken---they are not, and the pathology below is not a
criticism of them---but that the hybrid architecture requires a construct with
no established construction method, and the three natural candidates each fail
in a different direction.

\subsection{\texorpdfstring{The noise ceiling: minconf must be below $1-\varepsilon$}{The noise ceiling: minconf must be below 1-epsilon}}
\label{sec:noiseceiling}

\begin{proposition}[Confidence ceiling under label-flip noise]
\label{prop:ceiling}
Let the training labels be an independent symmetric label-flip realisation of
$P^{*}$ at rate $\varepsilon$, and let $g$ be a candidate rule whose pattern is
perfectly correct in $P^{*}$. Let $n=|\mathrm{supp}(g)|$ be the number of log
records matching $g$'s condition (a function of attributes only, hence
unaffected by label noise) and $\mathrm{conf}(g)$ its empirical confidence on
the noisy log. Then $n\cdot\mathrm{conf}(g)\sim
\mathrm{Binomial}(n,1-\varepsilon)$, so
$\mathbb{E}[\mathrm{conf}(g)]=1-\varepsilon$ exactly, and for any threshold
$\theta>1-\varepsilon$,
\begin{equation}
\Pr[\mathrm{conf}(g)\geq\theta]\ \leq\
\exp\!\left(-2n\bigl(\theta-(1-\varepsilon)\bigr)^{2}\right).
\label{eq:hoeffding}
\end{equation}
\end{proposition}
\begin{proof}
Each matching record's observed label agrees with the rule's effect iff it is not
flipped, independently with probability $1-\varepsilon$, giving the binomial;
\eqref{eq:hoeffding} is Hoeffding's inequality on the upper tail.
\end{proof}

Two consequences follow, and the second is the operationally important one.

\textbf{(i) Exclusion is probabilistic and sharpens exponentially in support.}
The threshold is not a strict finite-sample bar---by chance a rule can have fewer
than $\varepsilon n$ flipped labels and clear it---but the escape probability
vanishes as evidence accumulates. Collecting more logs makes recovery \emph{less}
likely, which inverts the usual intuition about data volume and is why the
symptom is easily misdiagnosed as data sparsity.

\textbf{(ii) Above the ceiling, the threshold is a low-support filter admitting
flukes.} Because \eqref{eq:hoeffding} decreases in $n$, the correct rules that
\emph{do} survive a mis-set minconf are exactly the low-support ones, surviving
by chance rather than by merit. At $\varepsilon=0.15$ and $\theta=0.90$ the exact
binomial tail is $0.191$ at support 36, $0.025$ at 200 and $2.0\times10^{-6}$ at
1{,}000. On the E-Document sample mined in Section~\ref{sec:ranking}, the median
permit-effect candidate has support 13 (tail $0.40$) and the median deny-effect
candidate support 412 (tail $1.7\times10^{-3}$), and \eqref{eq:hoeffding} bounds
the tail below $10^{-433}$ at the largest deny-effect support, 199{,}735. The
noise ceiling and the ranking pathology of Section~\ref{sec:ranking} therefore
act on the same rule pool, in opposite directions: the confidence stage retains
low-support flukes, and a support-ranked cap discards exactly those.

\textbf{Prescription.} Any noisy-log mining protocol should treat
$\mathrm{minconf}<1-\hat\varepsilon$ as a hard constraint, report
$\hat\varepsilon$ alongside minconf, and report the support distribution of the
retained rules. We make no claim about how often these settings co-occur in
practice, having no survey evidence; the statement is conditional on a pipeline
configured with $\mathrm{minconf}=0.9$ against a log with $\varepsilon\geq0.10$.
All runs reported here use $\mathrm{minconf}=0.60<1-\varepsilon$.

\subsection{Where the ceiling does not transfer}
\label{sec:noiserobust}
Proposition~\ref{prop:ceiling} is exact under \emph{independent, symmetric}
label-flip noise. Three common departures are not covered: class-dependent noise,
attribute- or subgroup-dependent noise, and dependent noise where flips correlate
within a user, resource or session. We tested whether the two consequences
survive them. A reference set of 400 latently-correct rules (clean confidence
exactly $1.0$, median support 142) was mined from the E-Document relation and
their confidences recomputed under four corruption processes sharing the nominal
marginal rate $\varepsilon=0.15$, over 10 seeds each (Table~\ref{tab:noise}).

\begin{table}[t]
\centering
\caption{Structured-noise robustness of Proposition~\ref{prop:ceiling}'s two
consequences.}
\vspace{-2pt}
{\footnotesize\raggedright E-Document; 400 latently-correct reference rules;
nominal $\varepsilon=0.15$; 10 seeds. ``Retain@.90'' is the fraction of correct
rules clearing $\mathrm{minconf}=0.90$; ``supp.\ shift'' is median survivor
support over median reference support, where a value below 1 is the low-support
bias of consequence~(ii). The proposition itself is unaffected: it is proved
under the symmetric IID model, whose row matches its prediction ($0.85$) to
within $0.001$.\par}
\vspace{2pt}
\label{tab:noise}
\footnotesize
\begin{tabular}{@{}lcccc@{}}
\toprule
noise process & mean conf. & sd & retain@.90 & supp.\ shift \\
\midrule
symmetric IID$^{*}$ & 0.8491 & 0.0028 & 0.071 & 0.28 \\
class-dependent$^{\dagger}$ & 0.8494 & 0.0028 & 0.073 & 0.29 \\
attribute-dependent & 0.8561 & 0.0020 & 0.158 & 0.60 \\
block-correlated & 0.8541 & 0.0180 & 0.299 & 0.75 \\
\bottomrule
\multicolumn{5}{@{}l@{}}{\footnotesize $^{*}$ the model Prop.~\ref{prop:ceiling} assumes. $^{\dagger}$ nearly indistinguishable}\\
\multicolumn{5}{@{}l@{}}{\footnotesize from symmetric here: the reference set is almost entirely deny-}\\
\multicolumn{5}{@{}l@{}}{\footnotesize effect at a $5.5\%$ permit base rate.}
\end{tabular}
\end{table}

The prediction is confirmed under its own model: mean observed confidence is
$0.8491$ against a predicted $0.85$. The exact ceiling does not transfer, and
its estimate becomes unreliable. Under attribute-dependent noise the mean rises
to $0.8561$ and the retained fraction from $0.071$ to $0.158$; under
block-correlated noise the seed-to-seed standard deviation is more than six
times larger ($0.0180$ against $0.0028$). Sweeping block granularity (per
subject, and subjects grouped by 8 and by 64) and the within-block flip
probability ($0.15$ to $1$) shows how far this goes. The seed-to-seed standard
deviation grows from $0.0019$ to as much as $0.1255$; at the coarsest
granularity with every flip confined to its block, mean confidence reaches
$0.9067$ and $72.7\%$ of correct rules clear $\mathrm{minconf}=0.90$, so
$1-\varepsilon$ is no longer a ceiling at all. Consequence~(ii) weakens with
correlation: the support-shift ratio rises from $0.29$ at the IID end to between
$0.72$ and $0.92$ once the within-block flip probability reaches $0.5$, at
every granularity, though not monotonically within one. Consequence~(ii) is
therefore an IID result and we state it as one: the prescription to report the
support distribution of retained rules stands, but the predicted
\emph{direction} of the bias should not be assumed when log noise is correlated
within users or sessions.

One test was less informative than intended. Our class-dependent process flips
deny$\to$permit at $\varepsilon$ and permit$\to$deny at $\varepsilon/4$; because
E-Document's permit rate is $5.5\%$ the reference set is overwhelmingly
deny-effect, for which this process is nearly indistinguishable from the
symmetric one, and the columns are indeed nearly identical. This condition does
not test asymmetry for this reference set, and a policy with a more balanced base
rate would be needed.

\subsection{Ranking, not the cap size, is the barrier}
\label{sec:ranking}
Below the noise ceiling a rule set mined by cap-and-rank can \emph{still} be
one-class, at every cap size. The natural inference is that grant patterns are absent from
the policy. They are not: they are excluded by a \emph{support-versus-cap
interaction}. Minority-class groups have low support while majority-class groups
reach very high support, so a support-ranked cap fills every slot with
majority-class rules before any minority group is reached. The effect is a
property of the \emph{ordering}, not of the budget: raising the cap does not help,
because the rules that would be admitted are never reached.

To separate the selection rule from any one implementation and from any one
policy's support distribution, we re-implemented the miner independently
(vectorised conjunction enumeration over the relation---over a uniform
200{,}000-row sample for the two large policies---with the
cross-attribute literals the ABAC rule language actually uses---$u.a=r.b$ and
$u.a\ni r.b$---since a hypothesis language of single-attribute equalities cannot
express relational rules and would make the measurement an artefact of the
language) and crossed four public policies with five candidate-selection rules
and three cap sizes, holding $\mathrm{min\_support}$ and $\mathrm{min\_conf}$
fixed within each policy.

The minority-support gap is present on all four policies. Median
majority-class candidate support exceeds median minority-class support by
$3.3\times$ (University, 40 vs 12) to $72.7\times$ (Workforce, 800 vs 11), and
minority candidates comprise between $0.05\%$ (Workforce, 9 of 17{,}140) and
$4.0\%$ (Project Management, 140 of 3{,}503) of all candidates. The precondition
the mechanism names holds throughout.

\begin{table}[t]
\centering
\caption{Selection-rule degeneracy over $4$ policies $\times$ $3$ cap sizes
($100/300/1000$), with an independently reimplemented miner. ``Degenerate'' means
the retained rule set is one-class on the covered region; ``minority'' is the
mean number of minority-class rules retained per cell and ``inf.\ mass'' the mean
informative-region mass. Per-policy configurations: E-Document at conjunction
cap $L{=}3$ over a uniform 200{,}000-row sample (min.\ support 5); University
and Project Management at $L{=}3$ over the full relation (min.\ support 3);
Workforce at $L{=}2$ over a 200{,}000-row sample (min.\ support 5); minimum
confidence $0.60$ throughout. Region and F1 figures are estimated on a
50{,}000-row subsample where the relation is larger.}
\label{tab:e3}
\footnotesize
\begin{tabular}{@{}lrrrr@{}}
\toprule
selection rule & degen. & rate & minority & inf.\ mass \\
\midrule
support                       & 12/12 & 1.00 &   0.0 & 0.00 \\
confidence (support tb) & 12/12 & 1.00 &   0.0 & 0.16 \\
confidence (low-support tb)   &  5/12 & 0.42 &   8.9 & 0.81 \\
lift                          &  3/12 & 0.25 & 152.5 & 0.31 \\
\textbf{class-balanced} & \textbf{0/12} & \textbf{0.00} & 90.6 & 0.16 \\
\bottomrule
\end{tabular}
\end{table}

Support-ranked selection degenerated in every one of its twelve cells, and
confidence-ranked selection also degenerated in every cell
(Table~\ref{tab:e3}), i.e.\ $24/24$ for the two commonest ranking keys. The
reason for the second is visible once stated: on these policies many
majority-class groups reach confidence $1.0$, so the confidence key ties and the
\emph{tiebreak} decides---and a support tiebreak reproduces support ranking
exactly. Breaking ties toward low support instead reduced degeneracy only to
$5/12$, and did so by collapsing the covered region almost entirely (mean
informative mass $0.81$), which is not a usable region map. Lift-ranking
degenerated in $3/12$ cells and did so in the opposite direction, selecting
all-permit rule sets. Only class-aware selection was non-degenerate
everywhere ($0/12$), recovering on average 91 minority-class rules per cell
while keeping informative mass at $0.16$. We implement it as confidence ordering
within each effect class followed by a $1{:}1$ interleave, which is the simplest
rule with the required property and is not tuned.

\emph{Implication for practice.} It is not sufficient to rank by confidence
rather than support; what is required is a selection rule that cannot let
majority-class support decide the cap, which confidence-first with a support
tiebreak still does. Practitioners should check the \emph{class composition} of
the retained rule set directly rather than trusting any ranking key.

\emph{Scope limits.} Workforce was mined at $L=2$ over a 200{,}000-row sample,
so its cells test the selection rule at a shallower hypothesis language than the
others. All four policies have skewed base rates ($2$--$5.5\%$ permit); whether
class-aware selection is still necessary, or still sufficient, on a policy with a
balanced base rate is untested. And the five rules are a search, not an
enumeration of the space of selection strategies.

\subsection{Consequences for reporting}
Covered-region quality $q$ and informative-region mass $m$ are the two quantities
every rules-versus-learner comparison implicitly conditions on, and both are
pipeline outputs. We adopt three reporting rules and recommend them generally.
\begin{enumerate}
\item \textbf{Report $q$ and $m$ only with the operator tuple} (conjunction cap
$L$, cap size, ranking criterion, minconf, and the region-map semantics---%
default-deny or abstaining) that produced them.
\item \textbf{Report the per-class decomposition of $q$, and a one-class flag.} A
bare macro-F1 carries no cross-dataset information at skewed base rates:
one-class prediction at a ${\sim}94$--$95\%$ base rate scores
$\approx(0.97+0)/2\approx0.486$ \emph{from either direction}. E-Document's
support-ranked rule set is all-deny and scores $q=0.486$; an all-permit rule set
on a log with the mirror-image base rate scores the same. Such an equality is an
artefact of the base rate.
\item \textbf{Never compare a bare $q$ across datasets.} Section~\ref{sec:c3}
gives the further reason: even a per-class-decomposed $q$ is not the operative
quantity.
\end{enumerate}

\section{C3: The Routing Decision}
\label{sec:c3}

\subsection{An exact decomposition}
Let $C=\{x:\mathrm{region}(x)=\text{covered}\}$ and $I=\mathcal{X}\setminus C$,
with $m=\Pr[x\in I]$. Let $\ell$ be any additive per-request loss---the
false-non-grant (FNG) indicator, the false-permit (FP) indicator, one minus
per-request accuracy---so aggregate risk decomposes over the partition. The
\emph{hybrid} arm routes $C$ to $R$ and $I$ to $L$; the \emph{learner} arm uses
$L$ everywhere; both are shielded identically.

\begin{proposition}[Exact hybrid decomposition]
\label{prop:decomp}
\begin{equation}
\mathrm{Risk}(\text{hybrid}+\sigma)-\mathrm{Risk}(\text{learner}+\sigma)
=(1-m)\bigl(\rC-\lC\bigr),
\label{eq:decomp}
\end{equation}
where $\rC=\mathbb{E}[\ell(\sigma\circ R)\mid C]$ and
$\lC=\mathbb{E}[\ell(\sigma\circ L)\mid C]$.
\end{proposition}
\begin{proof}
Both arms apply $\sigma\circ L$ on $I$, so the informative-region terms are
identical and cancel; the covered-region terms differ only in which decision
function is applied, and carry weight $1-m$.
\end{proof}

The proof is one line and we claim no depth for it; it uses nothing about
$\sigma$ beyond its being applied identically to both arms, so it holds equally
with no shield. Its value lies in what it forces a reporting protocol to
contain, which the rest of this section develops.

\emph{Scope: additive losses only.} \eqref{eq:decomp} requires the loss to be
additive across requests. FP and FNG are; macro-F1 is not, because it is
a non-linear function of per-class confusion counts and its value on a union of
regions is not the mass-weighted average of its values on the parts. We use
additive losses for every theorem-backed statement and report macro-F1 only
descriptively.

\emph{Conditional rates require the conditional mass.} FNG and FP are
conditional rates, not unconditional risks. The conditioning set is a property of
$Y$ alone and is therefore identical across arms, so the identity survives, but
the mass factor becomes the covered mass \emph{within that set}:
\begin{align}
\Delta\mathrm{FNG} &= (1-m_{G})(\rC^{G}-\lC^{G}), & G&=\{\Grant(Y)\},
\label{eq:decompG}\\
\Delta\mathrm{FP}  &= (1-m_{N})(\rC^{N}-\lC^{N}), & N&=\{\neg\Grant(Y)\}.
\label{eq:decompN}
\end{align}
This matters in practice. On E-Document under class-aware selection the three
masses differ sharply: the overall informative mass is $0.845$, but $m_{G}=0.257$
and $m_{N}=0.879$ (Table~\ref{tab:contrast}). A decomposition applied with the
overall $m$ would misattribute the false-non-grant effect by more than a factor
of four. Any report of a region-routing effect must state which subpopulation its
mass is conditioned on.

\subsection{The contrast is not one number}
\label{sec:twoloss}
Equations~\eqref{eq:decompG} and~\eqref{eq:decompN} are two different identities
over two different subpopulations, so \emph{the contrast is a different quantity
for each loss}. This is the point of C3, and it has a sharp consequence that a
single-loss report conceals:

\begin{quote}
\emph{A one-class rule set is trivially optimal on whichever loss its single
class cannot incur.} An all-deny rule set never over-grants, so $\rC^{N}=0$ and
its false-permit contrast is as favourable as it can be. An all-permit rule set
never withholds a grant, so $\rC^{G}=0$ and its false-non-grant contrast is as
favourable as it can be.
\end{quote}

A study that reports the contrast on one additive loss can therefore make a
degenerate rule set look like the right thing to route to, and can do so in
either direction depending on which loss it picks. Routing to rules is
justified on both losses only when both contrasts are negative, and we report
both throughout.
We know of no prior report of a rules-versus-learner routing effect that does so.

Three further consequences of \eqref{eq:decomp} are worth stating.

\emph{(i) Mass attenuates mechanically, but is not exogenous.} $1-m\geq0$, so
\emph{holding $\rC$ and $\lC$ fixed} the mass factor scales the effect and cannot
flip its sign. That much is algebra; it does not license the stronger reading
that mass is irrelevant, because $m$ is not exogenous to the contrast. Reshaping
the request distribution to raise $m$ can also starve the learner of
covered-region training data, raising $\lC$ while leaving $\rC$ unchanged. Mass
can therefore act through two channels, the mechanical $(1-m)$ factor and an
induced change in $\lC$, and only the first is sign-preserving. The common
justification for a hybrid architecture---``the rules leave a large region
unresolved, so add a learner''---is an argument about \emph{magnitude}, not
about \emph{direction}.

\emph{(ii) Covered-region rule quality alone does not sign it either.} The sign
is carried by the contrast. A rule set with $q=0.9$ should not be routed to if
the learner scores $0.95$ on the same requests.

\emph{(iii) The contrast is recoverable from published aggregates.} Since
\eqref{eq:decomp} is an identity and $m<1$, $\rC-\lC=\Delta/(1-m)$, so the
operative quantity can be read off any experiment reporting an aggregate effect
and a region mass, with its interval rescaled by $1/(1-m)$ when $m$ is treated
as fixed. We nonetheless
\emph{measure} $\rC$ and $\lC$ directly below, because the recovered quantity
coincides with the measured one only when the identity's conditioning
subpopulation is respected, and that is easy to get wrong.

\subsection{Measurement on a real policy: E-Document}
\label{sec:edoc}
\emph{Setup.} 20 seeds; per seed, disjoint 12{,}000-row training and
8{,}000-row test splits drawn uniformly from the 600{,}000-tuple enumeration;
label noise $\varepsilon=0.15$ on the training split only, clean $P^{*}$ labels
on test; identical test draws across selection rules within a seed;
gradient-boosted trees (\texttt{gbdt}, scikit-learn defaults) with
class-balanced sample weights and a $0.5$ threshold. Candidates are enumerated
at $L=3$ with minimum support 5 and minimum confidence $0.60$, and capped at
300 rules. The miner and the learner are fitted on the same evidence---both see
only the noisy training split; mining on the clean relation while the learner
fits noisy labels would hand the rules an advantage unrelated to routing.
Table~\ref{tab:contrast} compares $R$ and $L$ without a shield. Because the
E-Document hard-deny set admitted in Section~\ref{sec:c1} is sound, shielding
both arms with it would leave every false-non-grant term unchanged and would
set both arms' false permits on $H\cap C$ to zero.

\begin{table}[t]
\centering
\caption{Direct measurement of the decomposition terms on E-Document, on
both additive losses.}
\vspace{-2pt}
{\footnotesize\raggedright \texttt{gbdt} with class-balanced weights, 20 seeds,
identical test draws, 2000-resample paired bootstrap. $\rC$ and $\lC$ are
measured, not implied. A negative contrast favours routing to rules \emph{on that
loss}; the verdict column requires both. Note the two degenerate rows: the
all-deny \texttt{support} rule set attains the best possible false-permit
contrast ($\rC^{N}=0$) and the worst possible false-non-grant contrast
($\rC^{G}=1$), and the all-permit \texttt{lift} rule set does the reverse. Either
row can be made to look like a positive result by choosing a loss. Only
class-aware selection is favourable on both.\par}
\vspace{2pt}
\label{tab:contrast}
\footnotesize
\setlength{\tabcolsep}{2.9pt}\resizebox{\textwidth}{!}{\begin{tabular}{@{}lccccclccccll@{}}
\toprule
& \multicolumn{2}{c}{cov.-region F1} & \multicolumn{4}{c}{false-non-grant (on grant subpop.)}
& \multicolumn{4}{c}{false-permit (on non-grant subpop.)} & \\
\cmidrule(lr){2-3}\cmidrule(lr){4-7}\cmidrule(lr){8-11}
selection rule & deny & perm. & $m_{G}$ & $\rC^{G}$ & $\lC^{G}$ & $\rC^{G}\!-\!\lC^{G}$ [95\% CI]
& $m_{N}$ & $\rC^{N}$ & $\lC^{N}$ & $\rC^{N}\!-\!\lC^{N}$ [95\% CI] & verdict \\
\midrule
support \emph{(all-deny)} & 0.971 & 0.000 & 0.000 & 1.000 & 0.058 & $+0.942\ [+0.932,+0.950]$
                            & 0.000 & 0.000 & 0.303 & $-0.303\ [-0.312,-0.294]$ & metric-dep. \\
conf. (support tb) & 0.990 & 0.179 & 0.931 & 0.870 & 0.148 & $+0.721\ [+0.640,+0.796]$
                            & 0.839 & 0.000 & 0.188 & $-0.188\ [-0.197,-0.178]$ & metric-dep. \\
conf. (low-supp. tb) & 0.971 & 0.133 & 0.927 & 0.917 & 0.105 & $+0.812\ [+0.759,+0.859]$
                            & 0.930 & 0.003 & 0.333 & $-0.330\ [-0.348,-0.310]$ & metric-dep. \\
lift \emph{(all-permit)} & 0.000 & 0.766 & 0.122 & 0.000 & 0.048 & $-0.048\ [-0.060,-0.038]$
                            & 0.968 & 1.000 & 0.715 & $+0.285\ [+0.264,+0.308]$ & metric-dep. \\
\textbf{class-balanced}     & \textbf{0.931} & \textbf{0.852} & 0.257 & 0.011 & 0.044
                            & $\mathbf{-0.033\ [-0.049,-0.020]}$
                            & 0.879 & 0.125 & 0.240 & $\mathbf{-0.115\ [-0.127,-0.104]}$
                            & \textbf{rules} \\
\bottomrule
\end{tabular}}
\end{table}

\begin{figure}[t]
\centering
\includegraphics[width=0.72\linewidth]{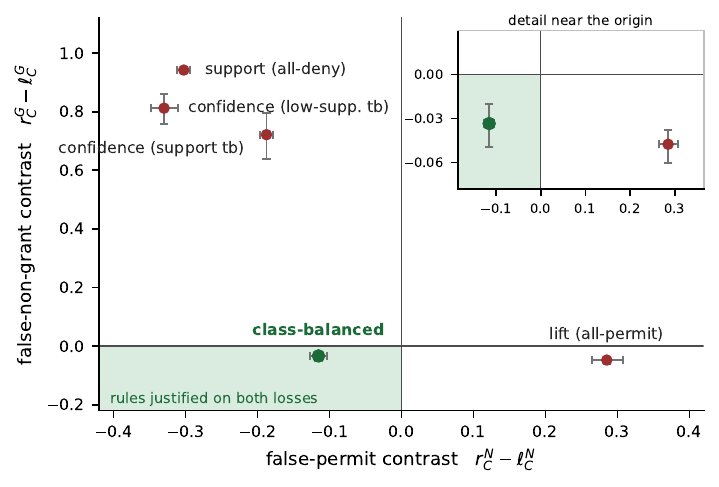}
\caption{The routing verdict is a point in two-loss space, not a sign
(E-Document, \texttt{gbdt} class-balanced, 20 seeds, 95\% paired-bootstrap
intervals; the data of Table~\ref{tab:contrast}). Routing the covered region to
rules is justified only in the shaded quadrant, where both contrasts are
negative. The one-class rule sets sit in opposite off-diagonal positions: an
all-deny rule set cannot over-grant, so its false-permit contrast is as
favourable as arithmetic allows, and an all-permit rule set cannot withhold a
grant, so its false-non-grant contrast is. Either can be presented as a positive
result by choosing a loss. Only class-aware selection lies in the shaded
quadrant; the inset magnifies the two points near the origin and shows that both
of its intervals exclude zero.}
\label{fig:twoloss}
\end{figure}

Four of five selection rules give opposite verdicts on the two losses
(Fig.~\ref{fig:twoloss}, Table~\ref{tab:contrast}). The all-deny support-ranked rule set has a
false-non-grant contrast of $+0.942$---on the region it claims to cover it denies
every request the policy grants---and simultaneously the best false-permit
contrast in the table, $-0.303$, for the trivial reason that it never grants at
all. The all-permit lift-ranked rule set is its mirror image, $-0.048$ on
false-non-grant and $+0.285$ on false-permit. A study reporting either loss alone
would recommend routing to a rule set that is one-class.

Only class-aware selection is favourable on both, with contrasts
$-0.033$ $[-0.049,-0.020]$ and $-0.115$ $[-0.127,-0.104]$, both intervals
excluding zero. This is the only configuration in our search for which the
routing decision is determined rather than metric-dependent, and it is the same
selection rule Section~\ref{sec:c2} finds non-degenerate everywhere.

As an implementation check, the measured effects reproduce
Proposition~\ref{prop:decomp} to floating-point precision: over all 200 rows
(5 selection rules $\times$ 2 losses $\times$ 20 seeds) the largest absolute
residual $\Delta_{\mathrm{obs}}-\Delta_{\mathrm{pred}}$ is $6.6\times10^{-17}$.

\subsection{Robustness of the verdict}

\emph{Label noise.} A miner aggregating over a support set attenuates label
noise while a learner fitting the same labels need not, so the rules should be
favoured more as noise rises---until noise also degrades the mined rules.
Table~\ref{tab:epssweep} sweeps $\varepsilon$ from $0$ to $0.25$ with 20 seeds
per level. Under class-aware selection the two-loss verdict is
\texttt{route-to-rules} for every $\varepsilon\in[0.02,0.15]$, with the
false-permit contrast growing in magnitude from $-0.043$ to $-0.121$ over that
range. It is not determined at either end. At $\varepsilon=0$, where the learner
sees clean labels, the false-permit contrast is $+0.004$ $[+0.003,+0.005]$---small,
but positive with an interval excluding zero---so the two losses disagree. At
$\varepsilon=0.25$ the false-non-grant interval, $[-0.025,+0.008]$, includes zero:
the noise now damages the rules as well (covered-region permit-F1 falls from
$0.93$ at $\varepsilon=0$ to $0.62$). The routing benefit on this policy is
therefore a noise-attenuation effect confined to a band of noise levels, not a
noise-free property of the rule set. A report on false-non-grant alone would
miss the lower end of that band: on that loss the contrast is negative even at
$\varepsilon=0$.

\begin{table}[t]
\centering
\caption{Sensitivity of the two-loss verdict to the label-noise level
(E-Document, \texttt{gbdt} class-balanced, 20 seeds per level; each level draws
its own splits, so the $\varepsilon=0.15$ rows differ slightly from
Table~\ref{tab:contrast}). The \texttt{support} rows are metric-dependent at
every level, in the direction Table~\ref{tab:contrast} explains. The
class-balanced rows favour rules on both losses for
$0.02\leq\varepsilon\leq0.15$; at $\varepsilon=0$ the losses disagree, and at
$\varepsilon=0.25$ the false-non-grant interval includes zero.}
\label{tab:epssweep}
\footnotesize
\setlength{\tabcolsep}{4pt}\begin{tabular}{@{}lcll@{}}
\toprule
$\varepsilon$ & rule & FNG contrast [95\% CI] & FP contrast [95\% CI] \\
\midrule
0.00 & support & $+0.976\ [+0.969,+0.981]$ & $-0.114\ [-0.119,-0.110]$ \\
     & class-bal. & $-0.010\ [-0.016,-0.006]$ & $+0.004\ [+0.003,+0.005]$ \\
0.02 & support & $+0.969\ [+0.964,+0.975]$ & $-0.182\ [-0.186,-0.178]$ \\
     & class-bal. & $-0.017\ [-0.022,-0.012]$ & $-0.043\ [-0.058,-0.031]$ \\
0.05 & support & $+0.962\ [+0.952,+0.970]$ & $-0.221\ [-0.225,-0.217]$ \\
     & class-bal. & $-0.023\ [-0.033,-0.015]$ & $-0.069\ [-0.078,-0.061]$ \\
0.10 & support & $+0.952\ [+0.946,+0.958]$ & $-0.268\ [-0.275,-0.262]$ \\
     & class-bal. & $-0.031\ [-0.037,-0.024]$ & $-0.100\ [-0.117,-0.084]$ \\
0.15 & support & $+0.943\ [+0.932,+0.952]$ & $-0.306\ [-0.313,-0.300]$ \\
     & class-bal. & $-0.028\ [-0.038,-0.019]$ & $-0.121\ [-0.140,-0.102]$ \\
0.25 & support & $+0.911\ [+0.902,+0.920]$ & $-0.375\ [-0.383,-0.367]$ \\
     & class-bal. & $-0.009\ [-0.025,+0.008]$ & $-0.066\ [-0.091,-0.039]$ \\
\bottomrule
\end{tabular}
\end{table}

\emph{Learner configuration.} A favourable contrast could equally mean the
learner was badly configured. The base rate is $5.5\%$, and a $0.5$ decision
threshold on such a problem is conservative, so the comparison must give the
learner its best configuration. We refit it in four ways---default, with
class-balanced sample weights, with a decision threshold tuned for macro-F1 on a
held-out (noisy) validation split, and with both---holding the mined rules fixed
so that any change is attributable to the learner alone
(Table~\ref{tab:fair}).

\begin{table}[t]
\centering
\caption{Is the routing verdict an artefact of an under-tuned learner?
E-Document, class-aware selection, $\varepsilon=0.15$, 20 seeds; the mined rules
are identical across rows ($\rC^{G}=0.012$ throughout), so all variation is the
learner's. ``tuned'' selects the decision threshold on a held-out validation
split (6{,}000 rows, drawn in addition to the training and test splits, so the
\texttt{balanced} row is close to but not identical with the class-balanced row
of Table~\ref{tab:contrast}). F1 is the learner's macro-F1; FNG is the
false-non-grant contrast. Only one of four configurations yields a
\texttt{route-to-rules} verdict; on false-non-grant alone every row favours the
rules.}
\label{tab:fair}
\footnotesize
\setlength{\tabcolsep}{3.5pt}\begin{tabular}{@{}lccrll@{}}
\toprule
learner & thr. & F1 & FNG & FP contrast [95\% CI] & verdict \\
\midrule
default & .500 & .587 & $-0.862$ & $+0.125\,[+.109,+.142]$ & metric-dep. \\
\textbf{balanced} & .500 & .547 & $\mathbf{-0.027}$ & $\mathbf{-0.108\,[-.126,-.091]}$ & \textbf{rules} \\
tuned & .255 & .709 & $-0.180$ & $+0.028\,[+.010,+.046]$ & metric-dep. \\
bal.+tuned & .582 & .699 & $-0.197$ & $+0.034\,[+.017,+.050]$ & metric-dep. \\
\bottomrule
\end{tabular}
\end{table}

The two-loss verdict does not survive threshold tuning. Under a class-balanced
learner at the default threshold,
routing to rules improves both losses. Under a threshold-tuned learner the
false-permit contrast becomes significantly \emph{positive} ($+0.028$
$[+0.010,+0.046]$ tuned, $+0.034$ $[+0.017,+0.050]$ balanced-and-tuned): the
learner, once permitted to move its threshold, over-grants less on the covered
region than the rules do, and routing there costs false permits while saving
false non-grants. The supported conclusion on this policy is therefore not
``routing to rules pays'' but:

\begin{quote}
On E-Document under class-aware selection, routing the covered region to rules
reduces false non-grants under every learner configuration tested, and reduces
false permits under only one of them. Which arm is preferable is undetermined
without an explicit relative price for the two errors, which is a deployment
input we do not hold.
\end{quote}

A report restricted to false-non-grant would have stated, truthfully and
misleadingly, that the sign is stable across all four learner configurations.
It is---on one loss. This is direct evidence for the reporting requirement of
Section~\ref{sec:twoloss}. More generally, because Proposition~\ref{prop:decomp}
makes the routing effect an identity in $m$, $\rC$ and $\lC$, a claim that
covered-region quality or informative mass ``governs'' hybrid benefit describes
a particular design until the decomposition terms themselves are reported; we
therefore recommend measuring $\rC$ and $\lC$ directly, as here.

\subsection{Replication on a second policy: Workforce}
\label{sec:workforce}
We repeated the measurement on Workforce, the other real-world-derived policy,
with the miner at $L=2$ (Section~\ref{sec:ranking}) and, because its permit base
rate is $2.0\%$, 50{,}000-row training and 20{,}000-row test splits; everything
else is as for E-Document (20 seeds, $\varepsilon=0.15$, class-balanced
\texttt{gbdt}, both rule set and learner fitted on the noisy training split).
Table~\ref{tab:workforce} reports the two contrasts.

\begin{table}[t]
\centering
\caption{The two-loss contrast on Workforce ($L=2$, 20 seeds,
$\varepsilon=0.15$, class-balanced \texttt{gbdt}). Every selection rule's
rules withhold grants that the learner makes (positive false-non-grant
contrast) and over-grant less than the learner (negative false-permit
contrast), with every interval excluding zero; no selection rule is favourable
on both losses. Lift and class-aware selection retain the same rule set
here.}
\label{tab:workforce}
\footnotesize
\setlength{\tabcolsep}{2.5pt}
\begin{tabular}{@{}lrll@{}}
\toprule
selection rule & $m_{G}$ & FNG contrast [95\% CI] & FP contrast [95\% CI] \\
\midrule
support (all-deny) & 0.000 & $+.969\,[+.962,+.975]$ & $-.198\,[-.206,-.190]$ \\
conf.\ (support tb) & 0.976 & $+.737\,[+.662,+.810]$ & $-.149\,[-.160,-.138]$ \\
conf.\ (low-supp.\ tb) & 0.975 & $+.738\,[+.661,+.813]$ & $-.150\,[-.160,-.139]$ \\
lift & 0.958 & $+.348\,[+.272,+.424]$ & $-.146\,[-.156,-.134]$ \\
class-balanced & 0.958 & $+.348\,[+.272,+.424]$ & $-.146\,[-.156,-.134]$ \\
\bottomrule
\end{tabular}
\end{table}

Every one of the five selection rules gives opposite verdicts on the two
losses. Routing Workforce's covered region to rules trades false non-grants
for false permits under every selection rule we tried, so the routing decision
on this policy is undetermined without a price for the two errors---and a study
reporting only false permits would recommend rules, while one reporting only
false non-grants would recommend the learner. Class-aware selection is again
non-degenerate (covered-region permit-F1 $0.48$), but the E-Document result that
it favours rules on both losses does not replicate. What replicates is the
divergence between the losses, which is the claim Section~\ref{sec:twoloss} makes.

\section{The Composed System}
\label{sec:eval}
The audits above examine the shield's inputs. This section measures what they
do inside a composed system: first on a real policy, with the hard-deny set that
C1 certified, and then on an authored synthetic policy that exercises the parts
of the mechanism no public benchmark supplies---a five-action space, a
break-glass branch and a residual surface. We label every zero as \emph{by
construction} (entailed by Lemma~\ref{lem:conform}) or \emph{observed}
(empirical, with denominator).

\subsection{On real policies}
\label{sec:eval-real}
We reuse the settings of Section~\ref{sec:edoc} (E-Document) and
Section~\ref{sec:workforce} (Workforce): class-balanced \texttt{gbdt},
$\varepsilon=0.15$, class-aware selection, 20 seeds. Each policy is shielded
with two hard-deny sets. The \emph{certified} set is the union of R1's clauses
for that policy that passed the soundness check; it covers
$4.2\%$ of E-Document's request space and
$15.7\%$ of Workforce's. The \emph{unchecked} set is
the union of all of R1's candidates for that policy, shipped without the check.
A label-conditional conformal gate (Section~\ref{sec:disjoint}) calibrated on a
disjoint 8{,}000-row split is included for comparison.

\begin{table}[t]
\centering
\caption{The composed system on two real policies (class-balanced \texttt{gbdt},
$\varepsilon=0.15$, 20 seeds; means). FP is the false-permit rate on non-granted
requests, FNG the false-non-grant rate on granted ones, and $H$-grant the grant
rate on requests in the certified set; zeros in that column for shielded arms
are by construction. Bootstrap intervals are in the artefact.}
\label{tab:realarms}
\footnotesize
\setlength{\tabcolsep}{2.4pt}
\begin{tabular}{@{}lrrrrrr@{}}
\toprule
 & \multicolumn{3}{c}{E-Document} & \multicolumn{3}{c}{Workforce}\\
\cmidrule(lr){2-4}\cmidrule(lr){5-7}
arm & FP & FNG & $H$-gr. & FP & FNG & $H$-gr.\\
\midrule
learner & 0.303 & 0.058 & 0.415 & 0.198 & 0.031 & 0.211 \\
\quad $+$ shield, certified $H$ & 0.284 & 0.058 & 0.000 & 0.164 & 0.031 & 0.000 \\
\quad $+$ shield, unchecked $H$ & 0.080 & 0.538 & 0.000 & 0.158 & 0.580 & 0.000 \\
\midrule
hybrid & 0.288 & 0.033 & 0.395 & 0.187 & 0.046 & 0.198 \\
\quad $+$ shield, certified $H$ & 0.271 & 0.033 & 0.000 & 0.155 & 0.046 & 0.000 \\
\quad $+$ shield, unchecked $H$ & 0.077 & 0.522 & 0.000 & 0.149 & 0.590 & 0.000 \\
\midrule
conformal@0.01 & 0.010 & 0.695 & 0.002 & 0.010 & 0.221 & 0.002 \\
\quad $+$ shield, certified $H$ & 0.010 & 0.695 & 0.000 & 0.010 & 0.221 & 0.000 \\
\bottomrule
\end{tabular}
\end{table}

Three things are visible in Table~\ref{tab:realarms}. First, the learner grants
on $41\%$ (E-Document) and $21\%$ (Workforce) of the requests the
certified set names; the shield removes every one of them and leaves the
false-non-grant rate unchanged to the last digit, because a sound $H$ contains no
granted request. Second, the unchecked set looks better on the loss a shield is
usually judged by---E-Document's false-permit rate falls from $0.30$ to
$0.08$---while it withholds $54\%$ and $58\%$ of everything the two
policies grant. On Workforce it buys almost nothing for that price: false permits
are $0.158$ against $0.164$ with the certified set. A study
reporting false permits alone would prefer the unchecked set on both policies.
Third, the conformal gate meets its false-permit target at $\alpha=0.01$ at a
false-non-grant rate of $0.70$ (E-Document) and $0.22$ (Workforce), and still grants on
a small fraction of the certified set, which the shield closes, as
Proposition~\ref{prop:disjoint} predicts.

\subsection{An authored policy with the full mechanism}
The remaining subsections use a policy we wrote, so nothing in them is
evidence about real policies; they demonstrate that the protocol is
implementable and that the propositions of Section~\ref{sec:theory} have the
consequences claimed.

\emph{Setup.} A latent policy of 15 rules (7 exceptions, 3 context modifiers)
in a critical-resource setting, with structural elements drawn from published
audit-trail and break-glass
requirements~\cite{brucker2009breakglass,rostad2006audit}. It is a
policy-derived benchmark, not an access log, and carries no clinical-deployment
claim. Three mutually disjoint samples per seed: training $n=6000$, calibration
$n=4000$, test $n=4000$; training labels are noisy realisations of $P^{*}$
($\varepsilon=0.15$), calibration and evaluation labels are $P^{*}$. Six
spoofable attributes are withheld from every learner and baseline. 20 seeds for
the arm comparison, 30 for the assumption-failure sweep. Mean per-seed test
denominators: ${\approx}2908$ non-grant, 1092 true-grant, 2781 hard-deny, 125
residual-surface.

\subsection{Report denial and abstention separately}
\label{sec:protocol}
Conditioning the cost side on the literal \textsc{deny} action admits a
degenerate optimum for any mechanism with an abstention. A selective predictor at
$\tau=0.95$ records literal denial $0.00005$ on true-grant requests while
diverting $0.757$ of them to \textsc{escalate}. That an abstaining mechanism can
drive a denial-conditioned metric to zero is definitional~\cite{chow1970}; we note
it because arms in this literature are routinely compared on the denial rate
alone. We therefore report $\mathrm{DNG}$ (denial on true-grant) and
$\mathrm{ENG}$ (escalation on true-grant) separately and combine them as
$C(\lambda)=\mathrm{DNG}+\lambda\,\mathrm{ENG}$.

$\lambda$ is a stipulated price, not a measured one, and the recommendation is
not invariant in it. Minimising $C(\lambda)$ subject to $\mathrm{FP}\leq0.02$
over all \texttt{gbdt} arms, the selected arm changes five times below
$\lambda\approx0.21$, walking down the abstention threshold:
\texttt{selective@0.95} on $[0,0.0023)$, \texttt{selective@0.9} on
$[0.0023,0.0281)$, \texttt{selective@0.8} on $[0.0281,0.1007)$,
\texttt{selective@0.7} on $[0.1007,0.1671)$, \texttt{selective@0.6} on
$[0.1671,0.2142)$, and \texttt{hybrid+shield} from $\lambda\approx0.214$ upward.
A paper that silently fixes $\lambda=1$ recommends the shielded hybrid; a
deployment in which an escalation resolves cheaply sits below the crossover and
should run a selective arm. Nothing in the accuracy data distinguishes these two
recommendations, and we hold no escalation-latency data, so we report the
components separately and decline to publish a single cost figure.

\subsection{The two mechanisms bound disjoint quantities}
\label{sec:disjoint}
The statistical mechanism we compose with the shield is a label-conditional
split-conformal gate (Fig.~\ref{alg:gate}). It is \emph{not} conformal risk
control in the sense of~\cite{angelopoulos2024crc}: the controlled quantity is a
single binary error probability conditional on the true label being non-granting,
so the classical split-conformal quantile argument applies directly. Because tree
ensembles produce tied scores, it is the $p$-value form, rather than a raw
quantile threshold, that makes the guarantee applicable here.

\begin{figure}[t]
\centering
\fbox{\parbox{0.90\linewidth}{\footnotesize
\textbf{Require:} model $f$; calibration set $\mathcal{D}_{\mathrm{cal}}$ disjoint
from training and test; target $\alpha$.\\
1. $s(x)\leftarrow\sum_{a:\Grant(a)}f(a\mid x)$\\
2. $S\leftarrow\{s(x_i):(x_i,y_i)\in\mathcal{D}_{\mathrm{cal}},\ \neg\Grant(y_i)\}$\\
3. $n\leftarrow|S|$ \hfill\emph{($n$ is the non-grant count)}\\
4. $p(x)\leftarrow\bigl(1+|\{s_i\in S: s_i\geq s(x)\}|\bigr)/(n+1)$
\hfill\emph{conservative}\\
5. \textbf{return} $\hat a(x)=\arg\max_{a:\Grant(a)}f(a\mid x)$ if
$p(x)\leq\alpha$, else \textsc{deny}
}}
\caption{The label-conditional conformal grant gate (tie-safe $p$-value form).
Training, calibration and test splits are mutually disjoint, and
hyperparameters are fixed a priori rather than tuned on calibration data.}
\label{alg:gate}
\end{figure}

\begin{proposition}[Disjointness, regime-indexed]
\label{prop:disjoint}
Let $H$ be sound, so the non-grant class is the disjoint union of $H$ and
$\Sres$, and let $\pi_H=\Pr[X\in H\mid\neg\Grant(Y)]\in(0,1)$.
(i)~\emph{Benign regime.} If a non-grant test score is exchangeable with the
non-grant calibration scores, the gate of Fig.~\ref{alg:gate} satisfies
$\Pr[\Grant(\hat a(X))\mid\neg\Grant(Y)]\leq\alpha$.
(ii)~The bound in (i) does not constrain grants on $H$: for every $\alpha>0$ it
is compatible with $\Pr[\Grant(\hat a(X))\mid X\in H]=\min(1,\alpha/\pi_H)>0$.
Conversely, $\sigma$ does not constrain grants on $\Sres$
(Lemma~\ref{lem:conform}(ii)). Neither mechanism implies the other, and
composing $\sigma$ after the gate retains both guarantees.
(iii)~\emph{Corrupted regime.} If the scores of corrupted inputs are not
exchangeable with the clean calibration scores, (i) need not hold, whereas
Proposition~\ref{prop:directional} still bounds grants on $H$.
\end{proposition}
\begin{proof}
(i) Under exchangeability the rank of the test score among the $n+1$ non-grant
scores is uniform, and counting ties against the test point
($s_i\geq s(x)$) makes $p(X)$ super-uniform, so
$\Pr[p(X)\leq\alpha]\leq\alpha$~\cite{vovk2005alrw}; the gate grants only when
$p(X)\leq\alpha$. (ii) The marginal rate decomposes as
$\pi_H g_H+(1-\pi_H)g_R$, where $g_H$ and $g_R$ are the grant rates on $H$ and
$\Sres$; any $g_H\leq\min(1,\alpha/\pi_H)$ with a correspondingly small $g_R$
satisfies (i). The statement about $\sigma$ is Lemma~\ref{lem:conform}(ii).
(iii) Proposition~\ref{prop:directional} assumes only (A3) and no
exchangeability.
\end{proof}

\begin{table}[t]
\centering
\caption{Key arms on the synthetic policy (\texttt{gbdt}, 20 seeds; mean per-seed
denominators 2908/1092/2781/125 for FP/cost/$H$/residual). $H$-grant of $0.000$
for shielded arms is by construction (Lemma~\ref{lem:conform}); other
zeros are observed. Every shielded and conformal arm in this table has
$\mathrm{esc}=0.0000$, so the ranking among those arms is invariant in
$\lambda$; only the selective arms are $\lambda$-sensitive.}
\label{tab:arms}
\footnotesize
\setlength{\tabcolsep}{3.6pt}\begin{tabular}{@{}lrrrrrr@{}}
\toprule
arm & FP & deny & esc & $C(1)$ & $H$-grant & resid. \\
\midrule
deny-by-default   & 0.0000 & 1.0000 & 0.0000 & 1.0000 & 0.0000 & 0.000 \\
rules             & 0.2060 & 0.2714 & 0.0017 & 0.2731 & 0.2231 & 0.000 \\
rules+shield      & 0.0000 & 0.2198 & 0.0000 & 0.2198 & 0.0000 & 0.000 \\
learner (raw)     & 0.0289 & 0.0836 & 0.0000 & 0.0836 & 0.0238 & 0.276 \\
learner+shield    & 0.0160 & 0.0947 & 0.0000 & 0.0947 & 0.0000 & 0.276 \\
hybrid            & 0.2310 & 0.0507 & 0.0017 & 0.0524 & 0.2405 & 0.268 \\
hybrid+shield     & 0.0156 & 0.0766 & 0.0000 & 0.0766 & 0.0000 & 0.268 \\
selective@0.7     & 0.0041 & 0.0164 & 0.3038 & 0.3202 & 0.0029 & 0.044 \\
selective@0.95    & 0.0000 & 0.0001 & 0.7567 & 0.7568 & 0.0000 & 0.000 \\
valthresh@0.7     & 0.0137 & 0.1188 & 0.0000 & 0.1188 & 0.0125 & 0.155 \\
conformal@0.005   & 0.0046 & 0.2010 & 0.0000 & 0.2010 & 0.0037 & 0.059 \\
\quad +shield     & 0.0028 & 0.1717 & 0.0000 & 0.1717 & 0.0000 & 0.059 \\
conformal@0.05    & 0.0472 & 0.0444 & 0.0000 & 0.0444 & 0.0458 & 0.391 \\
\quad +shield     & 0.0225 & 0.0889 & 0.0000 & 0.0889 & 0.0000 & 0.391 \\
\bottomrule
\end{tabular}
\end{table}

Both halves of Proposition~\ref{prop:disjoint} are visible in
Table~\ref{tab:arms}. No conformal level zeroes $H$: at $\alpha=0.05$ the gate
still grants on $0.0458$ of hard-deny requests and at $\alpha=0.005$ on
$0.0037$, and it grants on a non-zero fraction at all six levels tested, because
the guarantee is marginal over the non-grant class. A practitioner who reads
``distribution-free false-permit control'' as a guarantee that hard-deny
requests are never granted is therefore wrong. Conversely, the shield leaves the
residual exactly unchanged ($0.276$ shielded and unshielded; the paired per-seed
difference is exactly $0$ on every seed), whereas the gate reduces it to $0.113$
at $\alpha=0.01$ and $0.059$ at $\alpha=0.005$. The gate's realised false-permit
rate tracks its target across $\alpha\in\{0.005,\dots,0.2\}$ ($0.00455\pm0.00074$
at $0.005$ and $0.19938\pm0.00469$ at $0.2$, mean $\pm$ 95\% half-width over 20
seeds; $n\approx2908$ non-grant calibration points, so
$\alpha_{\min}=1/(n+1)=3.4\times10^{-4}$), so the bound is not bought with
needless conservatism. This is a consistency check, not evidence of validity,
which follows from Proposition~\ref{prop:disjoint}(i).

\subsection{Assumption failure}
Whatever the learner does---always permit, grant at random, or be trained on
poisoned labels or features---the shielded arm grants on no request in $H$. This
is \emph{entailed} by Lemma~\ref{lem:conform} under (A2) and (A3), not observed:
the learner's output is discarded on $H$. An attack on the learner therefore
shifts to the residual surface, where the shield is silent---which is the
operative security statement, and the reason C1 (naming more of the surface,
soundly) matters.

Input-side corruption is the interesting case (Table~\ref{tab:corruption}). In
every condition the realised unsafe-grant rate lies below $\epsm$, as
Proposition~\ref{prop:directional} requires; the bound is loose for an honest
learner and tight under the equality condition---substituting an always-permit
learner, the per-seed paired difference between the realised $H$-grant rate and
$\epsm$ is exactly zero on every seed. The availability control row is the
point of the table: it corrupts attributes so as to add spurious hard
denials, giving $\epsp=0.442$ and symmetric disagreement $0.135$ with
$\epsm=0.000$ and a realised unsafe-grant rate of exactly $0.000$. A symmetric
error term would suggest non-zero safety exposure where there provably is none.
The dominant failure mode is fail-open on missing attributes ($\epsm=0.598$).
This, not learner compromise, is where a deployment's exposure
lives, and it is an input-integrity problem rather than a machine-learning one.

\begin{table}[t]
\centering
\caption{Directional assumption-failure analysis; every column computed on the
same 30 seeds and test draws. $\epsm$ is the safety-relevant hard-deny
false-negative rate, $\epsp$ the availability-side false-positive rate,
$\varepsilon_{\mathrm{sym}}$ symmetric disagreement; the last two columns are the
realised grant rate on $H$ under an honest learner and under an always-permit
(adversarial) learner.}
\label{tab:corruption}
\footnotesize
\setlength{\tabcolsep}{3.6pt}\begin{tabular}{@{}lrrrrr@{}}
\toprule
corruption & $\epsm$ & $\epsp$ & $\varepsilon_{\mathrm{sym}}$ & honest $L$ & advers.\ $L$ \\
\midrule
abstraction            & 0.2680 & 0.0000 & 0.1862 & 0.0009 & 0.2680 \\
missing attributes     & 0.5976 & 0.0000 & 0.4151 & 0.0279 & 0.5976 \\
stale policy           & 0.1477 & 0.0000 & 0.1026 & 0.0292 & 0.1477 \\
availability (control) & 0.0000 & 0.4421 & 0.1350 & 0.0000 & 0.0000 \\
\bottomrule
\end{tabular}
\end{table}

\section{Threats to Validity}
\label{sec:threats}
Each threat states what it limits \emph{and} what survives it.

\textbf{C1 elicitors: one author and three language models of one family.}
The author elicitor saw the verdicts; the three independent elicitors saw only
the protocol's inputs, but they are models from a single provider, and model
elicitors may share failure modes that human analysts do not. They also saw the
schema as the harness presents it, with at most twelve example values per
attribute, which led some to name value families by suffix (the clause
language supports this). \emph{Limits:} the rates of
Table~\ref{tab:indep} as estimates for human analysts, and any comparison
between model sizes. \emph{Survives:} the agreement between four elicitors on
which sentences fail (15 restriction sentences unsound for every elicitor that
formalised them); the counterexamples themselves, which are verified facts about
the policies whoever wrote the clause; the attribution of those counterexamples
to other rules; and the cost of shipping an unchecked set. Because the
mechanism stratum is assigned by a fixed lexical rule, its membership cannot
have been tuned to the verdicts. The released harness accepts human elicitors
through the same freeze-then-certify interface, and a study with practising
policy authors is the natural next step.

\textbf{Five policies are not a sample.} ABAC Lab contains five enumerable
policies and we audited all five; the mechanism stratum is non-empty on two.
\emph{Limits:} any prevalence reading, and any claim that rejection rate tracks
policy scale or tenant structure. \emph{Survives:} the observation that the
stratum is empty exactly on the hand-authored instances and non-empty exactly on
the case-study-derived ones, which is a structural claim about documentation
style rather than a rate.

\textbf{C2's degeneracy result is about cap-and-rank, not about published
miners.} \emph{Limits:} it says nothing about the quality of Xu--Stoller,
Rhapsody or their successors, and Section~\ref{sec:reference} shows the reference
arm produces a usable policy where they are applicable. \emph{Survives:} the
observation that the hybrid architecture needs a two-class abstaining rule set
that, to our knowledge, none of them emits, and that three natural
constructions fail in three directions. The search covers five selection rules, three cap sizes and one
conjunction cap under one miner; it is a search, not an enumeration.

\textbf{The C3 verdict is configuration-dependent and we report it as such.}
\emph{Limits:} every numerical value of a contrast, and the routing
recommendation itself, which Table~\ref{tab:fair} shows survives on both losses
for one of four learner configurations, Table~\ref{tab:epssweep} only within a
band of noise levels, and Table~\ref{tab:workforce} not at all on Workforce.
\emph{Survives:} the identity, which is exact; the divergence of the two
losses, which appears for four of five selection rules on E-Document and for
all five on Workforce; and the reporting requirement, for which the instability
of the class-aware verdict is itself evidence.

\textbf{Region labels are definitional.} \emph{Covered} means the matching rules
agree on one effect. A different convention---treating conflict as covered under
a priority order, or adopting default-deny semantics as in
Section~\ref{sec:reference}---moves mass between regions and changes $m$.
\emph{Limits:} the numerical values of $m$, $\rC$ and $\lC$. \emph{Survives:}
Proposition~\ref{prop:decomp}, which holds for any partition into a rules-served
and a learner-served part.

\textbf{No deployment logs.} E-Document and Workforce are real-world-derived,
but their object models are generated, so splits are in-distribution rather
than cold-start, and the composed-system measurements of
Section~\ref{sec:eval-real} use label noise we inject. The rest of
Section~\ref{sec:eval} uses an authored policy. \emph{Limits:} every number in
Section~\ref{sec:eval} as an estimate for a deployed system. \emph{Survives:}
C1, C2 and C3, none of which depends on a log, and the qualitative contrast
between certified and unchecked deny sets, which follows from soundness alone.

\subsection{What this paper does not establish about system security}
\label{sec:securityscope}
This is a policy-level conformance and soundness analysis under stated trust
and mediation assumptions. It is not an end-to-end security evaluation.
The following are \emph{assumed}, not enforced, measured, or defended against:
complete mediation of the actuator path; communication integrity; freedom from
time-of-check-to-time-of-actuation races; replay freshness; policy-update
atomicity; availability under denial of service; and integrity of the attestation
infrastructure. A deployment in which any of these fails is outside every
guarantee stated here. Four properties are established or assumed at four
different strengths and we do not combine them into a claim that the system is
safe: the deterministic shield guarantee (proved, conditional on (A2) and (A3));
hard-deny-set soundness (checkable, and shown to fail on real artefacts; not
implied by any verified authorizer); attribute and input integrity (assumed by
(A2), dropped in Proposition~\ref{prop:directional}, and the dominant measured
failure mode at $\epsm=0.598$); and the conformal marginal guarantee
(distribution-free under exchangeability, marginal over a class, void under
corruption).

\section{Conclusion}
A learned access-control system is normally presented as a verified shield plus a
measured learner. Both components can be sound and the composed system still
wrong, because two artefacts between them are hand-built and unchecked: the set
the shield enforces, and the map that decides who answers which request.

\textbf{What is proved.} That a hard-deny-first shield admits no grant on the set
it names and is the identity on the residual surface
(Lemma~\ref{lem:conform})---so the shield is exactly as good as that set, and no
better. That the unsafe-grant rate on the set under corrupted inputs is bounded
by the directional false-negative rate, with an exact equality condition
(Proposition~\ref{prop:directional}). That a marginal class-conditional guarantee
is silent about a sub-population of that class, so a conformal gate and a shield
bound disjoint quantities and compose only in the benign regime
(Proposition~\ref{prop:disjoint}). That under independent symmetric label-flip
noise a correct rule's expected observed confidence is $1-\varepsilon$ and its
probability of clearing a higher threshold decays exponentially in support
(Proposition~\ref{prop:ceiling}). And that for additive losses the
hybrid-versus-learner effect equals covered mass times the contrast between rule
and learner risk on the covered region, with the mass conditioned on the loss's
own subpopulation (Proposition~\ref{prop:decomp}).

\textbf{What was observed.} On the two ABAC Lab policies whose prose states
restrictions, faithfully transcribed restriction clauses cover granted requests
at rates of $0.50$ to $0.70$ across an author and three blind language-model
elicitors, with up to 7{,}665 counterexamples per clause over relations of up to
794{,}250 tuples; $97\%$ of the counterexamples are granted by other rules, and
shipping any of the four candidate sets unchecked would withhold $16$--$78\%$ of
everything the policies grant. Vacuous clauses, which pass soundness while
constraining nothing, occur in independent elicitations. A published-style
permit-only miner produces a usable policy and no region map at all;
cap-and-rank over a mixed-effect pool is one-class in $24/24$ cells for the two
commonest ranking keys. On E-Document four of five selection rules give opposite
routing verdicts on the two additive losses, and the one favourable on both is
so under one of four learner configurations and only for
$0.02\leq\varepsilon\leq0.15$; on Workforce no selection rule is favourable on
both. On E-Document's composed system, a certified hard-deny set removes every
grant the learner makes on it at no availability cost.

\textbf{What we recommend.} Table~\ref{tab:checklist} collects the reporting
practices this paper argues for. None requires deployment data; each is computed
from artefacts that any study of learned access control already holds.

\begin{table}[t]
\centering
\caption{Reporting checklist for evaluations of learned access control behind a
trusted layer.}
\label{tab:checklist}
\footnotesize
\setlength{\tabcolsep}{3pt}
\begin{tabular}{@{}>{\raggedright\arraybackslash}p{2.6cm}>{\raggedright\arraybackslash}p{9.5cm}l@{}}
\toprule
artefact & report & see \\
\midrule
hard-deny set & soundness certificate, coverage of each clause, size of the enumerated relation & \S\ref{sec:c1} \\
elicitation & protocol, number of elicitors, provenance strata reported apart & \S\ref{sec:provenance}, \S\ref{sec:independent} \\
mining & $\mathrm{minconf}$ against $\hat\varepsilon$; support distribution of retained rules & \S\ref{sec:noiseceiling} \\
region map & semantics (total or abstaining), class composition, full operator tuple & \S\ref{sec:c2} \\
routing & $\rC-\lC$ on \emph{every} additive loss, with the learner configuration & \S\ref{sec:c3} \\
cost & denial and abstention separately; any price $\lambda$ stated & \S\ref{sec:protocol} \\
inputs & directional $\epsm$ and $\epsp$, not symmetric disagreement & \S\ref{sec:theory} \\
\bottomrule
\end{tabular}
\end{table}

\textbf{What remains unknown.} Whether human policy authors write unsound
restriction clauses at rates resembling those of the four elicitors here; the
harness makes that study cheap. Whether the mechanism's prevalence tracks policy
scale or tenant structure, which five policies cannot address. Whether
class-aware selection remains necessary, or sufficient, on a policy with a
balanced base rate. Whether any selection rule yields a routing verdict that is
robust to learner tuning on both losses. And everything in
Section~\ref{sec:securityscope}.

\section*{Data and Code Availability}
All policy artefacts used here are public: the five ABAC Lab
policies~\cite{bui2025abaclab}. The complete analysis code---policy parser and
enumerator, hard-deny clause language with provenance tagging and exhaustive
certifier, the multi-elicitor harness, the rule miner with its five selection
rules and the permit-cover reference arm, the four noise processes, the shield,
the conformal gate, the two-loss decomposition measurement, and the synthetic
composed-system pipeline---is provided to the reviewers as supplementary
material and will be archived with a persistent identifier on acceptance. It
includes the exact prompts shown to the independent elicitors, their raw
outputs, and the frozen candidate files with fingerprints and timestamps. A
single script regenerates every CSV from which a number in this manuscript is
quoted and then checks each quoted value against the regenerated results. The
parser is validated against each policy's published statistics at run time. All
experiments are deterministic given their seeds and run on a commodity laptop
(Python~3 with NumPy, SciPy and scikit-learn).

\section*{Acknowledgment}
The authors used a generative AI assistant for code implementation, result
verification and manuscript editing. All experimental results were produced by
the released code and verified by the authors, who take full responsibility for
the content of this article.

\appendix
\section{The author elicitor's candidate clauses}
\label{app:r1}

Every candidate written by the author elicitor R1, with its provenance stratum (Section~\ref{sec:provenance}), its exhaustive verdict, its coverage (requests covered) and its counterexample count (covered requests the policy grants). Guards use the clause language of the released artefact.

{\scriptsize
\begin{longtable}{@{}l l l r r >{\raggedright\arraybackslash}p{8.4cm}@{}}
\toprule ID & strat. & verdict & cov. & c/ex & clause \\ \midrule \endhead
\bottomrule \endfoot
HC1 & excl. & unsound & 68 & 9 & \texttt{action in {[}'addItem'{]} AND resource.type in HR AND NOT user.position in nurse} \\
HC2 & nec. & unsound & 68 & 9 & \texttt{action in {[}'addItem'{]} AND resource.type in HR AND NOT user.ward = resource.ward} \\
HC3 & nec. & unsound & 80 & 4 & \texttt{action in {[}'addNote'{]} AND resource.type in HR AND NOT user.uid = resource.patient} \\
HC4 & excl. & unsound & 240 & 6 & \texttt{action in {[}'read'{]} AND resource.type in HRitem AND NOT user.uid = resource.author} \\
HC5 & nec. & unsound & 236 & 8 & \texttt{action in {[}'read'{]} AND resource.type in HRitem AND NOT user.specialties > resource.topics} \\
HC6 & sil. & unsound & 48 & 4 & \texttt{action in {[}'read'{]} AND resource.type in HRitem AND user.position in nurse} \\
HC7 & nec. & unsound & 75 & 8 & \texttt{action in {[}'addItem'{]} AND resource.type in HR AND NOT user.teams {]} resource.treatingTeam} \\
HC8 & sil. & sound & 84 & 0 & \texttt{action in {[}'read'{]} AND resource.type in HR} \\
UN1 & excl. & unsound & 240 & 4 & \texttt{action in {[}'read', 'write'{]} AND resource.type in roster AND NOT user.department in registrar} \\
UN2 & excl. & unsound & 200 & 20 & \texttt{action in {[}'read'{]} AND resource.type in transcript AND NOT user.department in registrar} \\
UN3 & excl. & sound & 480 & 0 & \texttt{action in {[}'read', 'setStatus'{]} AND resource.type in application AND NOT user.department in admissions} \\
UN4 & excl. & sound & 108 & 0 & \texttt{action in {[}'assignGrade'{]} AND resource.type in gradebook AND NOT user.position in faculty} \\
UN5 & nec. & sound & 244 & 0 & \texttt{action in {[}'assignGrade', 'changeScore'{]} AND resource.type in gradebook AND NOT user.crsTaught {]} resource.crs} \\
UN6 & nec. & sound & 252 & 0 & \texttt{action in {[}'checkStatus'{]} AND resource.type in application AND NOT user.uid = resource.student} \\
UN7 & excl. & unsound & 200 & 30 & \texttt{action in {[}'read'{]} AND resource.type in transcript AND NOT user.isChair in True} \\
UN8 & nec. & sound & 120 & 0 & \texttt{action in {[}'readMyScores'{]} AND resource.type in gradebook AND NOT user.crsTaken {]} resource.crs} \\
PM1 & excl. & sound & 144 & 0 & \texttt{action in {[}'write'{]} AND resource.type in \{budget schedule\} AND NOT user.projectsLed {]} resource.project} \\
PM2 & excl. & sound & 144 & 0 & \texttt{action in {[}'read', 'write'{]} AND resource.type in budget AND NOT user.projectsLed {]} resource.project} \\
PM3 & nec. & sound & 592 & 0 & \texttt{action in {[}'setStatus'{]} AND resource.type in task AND NOT user.tasks {]} resource.rid} \\
PM4 & nec. & unsound & 608 & 16 & \texttt{action in {[}'read', 'request'{]} AND resource.type in task AND resource.proprietary in True} \\
PM5 & nec. & sound & 480 & 0 & \texttt{action in {[}'request'{]} AND resource.type in task AND NOT user.expertise > resource.expertise} \\
PM6 & excl. & unsound & 480 & 8 & \texttt{action in {[}'request'{]} AND resource.type in task AND NOT user.isEmployee in True} \\
ED1 & expl. & unsound & 61,932 & 6,489 & \texttt{action in {[}'search', 'send', 'view'{]} AND resource.type in invoice AND NOT user.department endswith Sales} \\
ED2 & expl. & unsound & 53,350 & 1,706 & \texttt{action in {[}'send', 'view'{]} AND resource.type in bankingNote AND NOT user.department endswith ICT} \\
ED3 & expl. & unsound & 27,416 & 1,774 & \texttt{action in {[}'send', 'view'{]} AND resource.type in paycheck AND NOT user.payrollingPermissions in True} \\
ED4 & expl. & unsound & 18,262 & 782 & \texttt{action in {[}'send'{]} AND resource.type in salesOffer AND NOT user.department endswith Sales} \\
ED5 & expl. & unsound & 26,488 & 835 & \texttt{action in {[}'view'{]} AND resource.type in trafficFine AND NOT user.department endswith CustomerCare} \\
ED6 & expl. & unsound & 18,945 & 495 & \texttt{action in {[}'send'{]} AND resource.type in contract AND NOT user.department endswith HR} \\
ED7 & expl. & unsound & 25,532 & 3,681 & \texttt{action in {[}'view'{]} AND resource.type in invoice AND NOT user.department endswith Secretary} \\
ED8 & expl. & unsound & 35,100 & 543 & \texttt{action in {[}'send'{]} AND user.registered in False} \\
ED9 & rep. & sound & 12,000 & 0 & \texttt{action in {[}'send'{]} AND user.registered in False AND user.role in customer} \\
ED10 & nec. & sound & 5,580 & 0 & \texttt{action in {[}'view'{]} AND user.role in admin AND resource.isConfidential in True} \\
ED11 & expl. & unsound & 3,672 & 942 & \texttt{action in {[}'view'{]} AND user.department endswith Audit AND resource.containsPersonalInfo in True} \\
ED12 & expl. & sound & 4,800 & 0 & \texttt{action in {[}'send'{]} AND user.tenant in privateReceiver} \\
ED13 & expl. & sound & 7,769 & 0 & \texttt{action in {[}'view'{]} AND user.role in helpdesk AND NOT user.tenant = resource.tenant} \\
ED14 & expl. & unsound & 16,500 & 2,432 & \texttt{action in {[}'send'{]} AND user.tenant in largeBank AND NOT user.position in \{officeManager seniorOfficeManager\}} \\
ED15 & expl. & unsound & 22,200 & 559 & \texttt{action in {[}'view'{]} AND user.tenant in londonOffice AND NOT user.position in \{director secretary\}} \\
ED16 & expl. & unsound & 18,900 & 837 & \texttt{action in {[}'view'{]} AND user.tenant in reseller AND NOT user.department endswith Customer} \\
ED17 & excl. & unsound & 141,000 & 605 & \texttt{action in {[}'readMetaInfo'{]} AND NOT user.role in helpdesk} \\
WF1 & expl. & unsound & 6,400 & 20 & \texttt{action in {[}'receive'{]} AND resource.type in stockRefillRequest AND NOT user.group in provisioning} \\
WF2 & expl. & sound & 6,840 & 0 & \texttt{action in {[}'receive'{]} AND resource.type in stockRefillRequest AND NOT user.group in techSupport} \\
WF3 & expl. & sound & 77,100 & 0 & \texttt{action in {[}'complete', 'view'{]} AND resource.type in task AND NOT user.department in workforce} \\
WF4 & expl. & unsound & 19,500 & 7,665 & \texttt{action in {[}'complete', 'view'{]} AND resource.type in task AND user.position in technician AND NOT user.uid = resource.assignedTechnician} \\
WF5 & expl. & sound & 28,800 & 0 & \texttt{action in {[}'createOneTimeWorkOrder', 'createRecurrentWorkOrder'{]} AND resource.type in workOrder AND NOT user.position in helpdeskOperator} \\
WF6 & expl. & unsound & 40,242 & 2,613 & \texttt{action in {[}'complete', 'view'{]} AND resource.type in task AND resource.contractStatus in inactive} \\
WF7 & expl. & sound & 2,800 & 0 & \texttt{action in {[}'complete', 'view'{]} AND resource.type in resourceRequest AND user.position in warehouseOperator AND NOT user.uid = resource.assignedEmployee} \\
WF8 & excl. & sound & 6,460 & 0 & \texttt{action in {[}'markComplete'{]} AND resource.type in stockRefillRequest AND NOT user.position in warehouseManager} \\
WF9 & expl. & sound & 3,000 & 0 & \texttt{action in {[}'delete', 'modify'{]} AND resource.type in workOrder AND user.position in helpdeskManager} \\
WF10 & expl. & unsound & 7,500 & 20 & \texttt{action in {[}'createRecurrentWorkOrder'{]} AND user.position in salesManager} \\
\end{longtable}}

\section{Verdicts per restriction sentence and elicitor}
\label{app:rules}

The 31 rule sentences that the lexical rule marks as explicit restrictions, and the verdict of each elicitor's clause(s) read from that sentence: U unsound, S sound, V vacuous, -- not formalised. Where an elicitor wrote several clauses from one sentence, U is reported if any is unsound. Models: L1~Claude Opus 5.5, L2~Claude Sonnet 5.5, L3~Claude Haiku 5.5.

{\small
\begin{longtable}{@{}l r >{\raggedright\arraybackslash}p{6.6cm} c c c c@{}}
\toprule policy & rule & section / subsection & R1 & L1 & L2 & L3 \\ \midrule \endhead
\bottomrule \endfoot
edocument & 1 & eDocs / General & U & U & U & U \\
edocument & 3 & eDocs / Helpdesk & S & S & -- & -- \\
edocument & 7 & Large Bank / Invoices & U & U & U & U \\
edocument & 8 & Large Bank / Banking Notes & U & U & U & U \\
edocument & 9 & Large Bank / Paychecks & U & U & U & U \\
edocument & 10 & Large Bank / Sales Offers & U & U & U & U \\
edocument & 11 & Large Bank / Internal Communication with Local Bank Offices & U & U & U & U \\
edocument & 12 & Large Bank / Audits & U & -- & -- & U \\
edocument & 13 & Large Bank Leasing / Traffic Fines & U & U & U & -- \\
edocument & 14 & Large Bank Leasing / Invoices & -- & U & U & -- \\
edocument & 15 & Local Bank Offices / Applicable Policies & U & U & U & U \\
edocument & 17 & ICT Provider / Invoices & U & U & U & U \\
edocument & 19 & Europe Region / Contracts & U & U & U & U \\
edocument & 23 & Reseller / Subtenant Document Viewing & U & -- & -- & -- \\
edocument & 25 & Registered Private Receivers / Applicable Policies & S & V & V & V \\
workforce & 1 & eWorkforce / eWorkforce Helpdesk & S & S & S & S \\
workforce & 2 & eWorkforce / eWorkforce Helpdesk & -- & S & U & U \\
workforce & 5 & eWorkforce / Internal Workforce & S & S & S & S \\
workforce & 6 & eWorkforce / Internal Workforce & U & U & U & U \\
workforce & 7 & eWorkforce / Internal Workforce & U & -- & -- & -- \\
workforce & 10 & eWorkforce / Internal Warehouse & S & S & S & S \\
workforce & 12 & PowerProtection / Sales Managers & U & S & S & U \\
workforce & 15 & PowerProtection / Stock Refill Notifications & U & S & S & S \\
workforce & 17 & TelCo / Work Orders & -- & V & V & U \\
workforce & 18 & TelCo / Work Orders & -- & S & U & U \\
workforce & 19 & TelCo / Stock Refill Notifications & S & S & S & S \\
workforce & 22 & Helpdesk Supplier / Helpdesk Managers & S & S & S & U \\
workforce & 23 & Workforce Supplier / General & -- & U & U & U \\
workforce & 24 & Workforce Supplier / Technicians & -- & U & U & U \\
workforce & 26 & Subcontractor / General & -- & -- & -- & -- \\
workforce & 27 & Subcontractor / Technicians & -- & -- & -- & -- \\
\end{longtable}}

\section{Elicitation protocol}
\label{app:protocol}

Each independent elicitor was an isolated agent with read access to exactly five files, one per policy, each consisting of the system prompt below followed by a user message holding the policy's prose packet (formal rule blocks stripped and leak-checked), its attribute schema (at most twelve example values per attribute), and the output format. The files are released verbatim. The agents were instructed not to open, list or search any other file, not to use the web, and to report every file read; each reported reading only the five files. One agent wrote its output files with a shell command rather than the file-writing tool; none read anything else. Outputs were parsed and frozen (per-clause fingerprint, UTC timestamp) before any certification.

\paragraph{System prompt (verbatim).}
\begin{quote}\small\ttfamily\raggedright
You are an access-control analyst. You are given a policy's PROSE description and its OBJECT MODEL SCHEMA. You have NOT been given the formal rules, and you must not guess at them.

\noindent Your task: write candidate HARD-DENY clauses. A hard-deny clause is a categorical prohibition that a trusted enforcement layer will apply before any permitting logic: if a request matches the clause, it is denied outright.

\noindent Read every sentence that states or implies a categorical prohibition -- including sentences of the form "only members of X may VERB", which imply that everyone who is not in X is denied VERB. Write one candidate per such sentence.

\noindent Be faithful to the sentence. Do not add guards the sentence does not state, and do not weaken it. Use ONLY attribute names and values that appear in the schema.
\end{quote}

\section{Counterexample attribution by elicitor}
\label{app:attribution}

Section~\ref{sec:attribution}, broken down by elicitor. Fractions are of counterexamples, over the unsound mechanism-stratum clauses whose source rule can be located; ``scoped repair'' counts clauses made sound by a single guard naming the tenant of the clause's own section.

\begin{table}[h]\centering\small\begin{tabular}{@{}lrrrrrc@{}}\toprule
elicitor & clauses & c/ex & own rule & other subsect. & other tenant & scoped repair\\ \midrule
R1 & 16 & 31,393 & 0.000 & 0.783 & 0.592 & 3/16\\
L1 & 17 & 17,192 & 0.034 & 0.819 & 0.661 & 1/18\\
L2 & 22 & 25,939 & 0.023 & 0.878 & 0.552 & 3/22\\
L3 & 22 & 62,930 & 0.055 & 0.905 & 0.661 & 4/23\\
\bottomrule\end{tabular}\end{table}

\section{Composed system on real policies, with intervals}
\label{app:realarms}

Table~\ref{tab:realarms} with 95\% paired-bootstrap intervals over 20 seeds.

\begin{table}[h]\centering\footnotesize\begin{tabular}{@{}lccc@{}}\toprule
E-Document arm & FP & FNG & $H$-grant\\ \midrule
learner & 0.303 [0.294, 0.312] & 0.058 [0.050, 0.068] & 0.415 [0.377, 0.452]\\
learner+shield & 0.284 [0.276, 0.292] & 0.058 [0.050, 0.068] & 0.000 [0.000, 0.000]\\
learner+shield(unchecked H) & 0.080 [0.077, 0.083] & 0.538 [0.528, 0.549] & 0.000 [0.000, 0.000]\\
hybrid & 0.288 [0.279, 0.298] & 0.033 [0.028, 0.040] & 0.395 [0.359, 0.431]\\
hybrid+shield & 0.271 [0.262, 0.279] & 0.033 [0.028, 0.040] & 0.000 [0.000, 0.000]\\
hybrid+shield(unchecked H) & 0.077 [0.074, 0.081] & 0.522 [0.513, 0.534] & 0.000 [0.000, 0.000]\\
conformal@0.01 & 0.010 [0.009, 0.010] & 0.695 [0.665, 0.723] & 0.002 [0.000, 0.005]\\
conformal@0.01+shield & 0.010 [0.009, 0.010] & 0.695 [0.665, 0.723] & 0.000 [0.000, 0.000]\\
\bottomrule\end{tabular}\end{table}

\begin{table}[h]\centering\footnotesize\begin{tabular}{@{}lccc@{}}\toprule
Workforce arm & FP & FNG & $H$-grant\\ \midrule
learner & 0.198 [0.190, 0.206] & 0.031 [0.025, 0.038] & 0.211 [0.187, 0.236]\\
learner+shield & 0.164 [0.156, 0.173] & 0.031 [0.025, 0.038] & 0.000 [0.000, 0.000]\\
learner+shield(unchecked H) & 0.158 [0.150, 0.167] & 0.580 [0.568, 0.593] & 0.000 [0.000, 0.000]\\
hybrid & 0.187 [0.179, 0.194] & 0.046 [0.038, 0.053] & 0.198 [0.175, 0.223]\\
hybrid+shield & 0.155 [0.147, 0.163] & 0.046 [0.038, 0.053] & 0.000 [0.000, 0.000]\\
hybrid+shield(unchecked H) & 0.149 [0.141, 0.157] & 0.590 [0.577, 0.603] & 0.000 [0.000, 0.000]\\
conformal@0.01 & 0.010 [0.009, 0.011] & 0.221 [0.208, 0.237] & 0.002 [0.001, 0.002]\\
conformal@0.01+shield & 0.010 [0.009, 0.010] & 0.221 [0.208, 0.237] & 0.000 [0.000, 0.000]\\
\bottomrule\end{tabular}\end{table}

\section{Selected arm as a function of the escalation price}
\label{app:lambda}

Section~\ref{sec:protocol}: the arm minimising $C(\lambda)$ subject to $\mathrm{FP}\leq0.02$ over all \texttt{gbdt} arms of the synthetic policy.

\begin{table}[h]\centering\small\begin{tabular}{@{}lcccc@{}}\toprule
$\lambda$ interval & arm & FP & DNG & ENG\\ \midrule
$[0.0000, 0.0023)$ & \texttt{selective@0.95} & 0.0000 & 0.0000 & 0.7567\\
$[0.0023, 0.0281)$ & \texttt{selective@0.9} & 0.0001 & 0.0005 & 0.5765\\
$[0.0281, 0.1007)$ & \texttt{selective@0.8} & 0.0013 & 0.0049 & 0.4175\\
$[0.1007, 0.1671)$ & \texttt{selective@0.7} & 0.0041 & 0.0164 & 0.3038\\
$[0.1671, 0.2142)$ & \texttt{selective@0.6} & 0.0094 & 0.0336 & 0.2004\\
$[0.2142, 1.0000)$ & \texttt{hybrid+shield} & 0.0156 & 0.0766 & 0.0000\\
\bottomrule\end{tabular}\end{table}

\section{Shield conformance tests}
\label{app:conformance}

Implementation checks of the synthetic composed system (\url{synthetic/test_shield_conformance.py}). Each is exhaustive over the stated cases; all pass. Checks T2, T3, T5 and T8 concern the synthetic policy's break-glass branch, which Section~\ref{sec:eval} uses only as part of the demonstration.

\begin{table}[h]\centering\small\begin{tabular}{@{}l>{\raggedright\arraybackslash}p{7.6cm}>{\raggedright\arraybackslash}p{4.4cm}l@{}}\toprule
test & property checked & cases & result\\ \midrule
T1 & no action of the shield grants on $H$ (Lemma~\ref{lem:conform}(i)) (attested=False) & 5578 H-states x 5 actions = 27890 cases & PASS\\
T1 & no action of the shield grants on $H$ (Lemma~\ref{lem:conform}(i)) (attested=True) & 5578 H-states x 5 actions = 27890 cases & PASS\\
T2 & break-glass set disjoint from $H$ (attested=False) & n=8000 & PASS\\
T2 & break-glass set disjoint from $H$ (attested=True) & n=8000 & PASS\\
T3 & break-glass set contains no request $P^{*}$ refuses (attested=False) & n=8000, violations=0 & PASS\\
T3 & break-glass set contains no request $P^{*}$ refuses (attested=True) & n=8000, violations=0 & PASS\\
T4 & shield is the identity on $\Sres$ (Lemma~\ref{lem:conform}(ii)) (attested=False) & 1260 cases & PASS\\
T4 & shield is the identity on $\Sres$ (Lemma~\ref{lem:conform}(ii)) (attested=True) & 1260 cases & PASS\\
T5 & shield passes every action through outside $H$ and the break-glass set  & 7995 cases & PASS\\
T6 & equality case of Proposition~\ref{prop:directional} under an always-grant learner  & eps\_H-=0.27468 realised=0.27468 & PASS\\
T6b & the bound of Proposition~\ref{prop:directional} excludes hard-deny false positives  & eps\_H+=0.00000 not in bound & PASS\\
T7 & shield placed after the conformal gate still denies on $H$  & n=4170 & PASS\\
T8 & every break-glass transformation is flagged for audit  & 244 transformations & PASS\\
\bottomrule\end{tabular}\end{table}

\section{Reproducing this paper}
\label{app:repro}

All results are produced by the released artefact (Python~3 with NumPy, SciPy, scikit-learn, pandas and matplotlib) from the public ABAC Lab policy files. \texttt{bash run\_all.sh} regenerates every result file in about 35 minutes on a 2014 laptop CPU (Intel Core i7-4610M) and then runs \texttt{experiments/check\_manuscript.py}, which re-derives every table row and every number quoted in the text from the regenerated files and fails on any disagreement. All experiments are deterministic given their seeds; on the authors' machine a full rerun reproduced every reported value.


\begin{thebibliography}{54}
\footnotesize
\bibitem{hu2014abac} V.~C. Hu \emph{et al.}, ``Guide to attribute based access control (ABAC) definition and considerations,'' Nat. Inst. Standards Technol., Gaithersburg, MD, USA, NIST Special Publication 800-162, 2014, doi: 10.6028/NIST.SP.800-162.
\bibitem{xu2015mining} Z.~Xu and S.~D. Stoller, ``Mining attribute-based access control policies,'' \emph{IEEE Trans. Dependable Secure Comput.}, vol.~12, no.~5, pp.~533--545, 2015.
\bibitem{cotrini2018sparse} C.~Cotrini, T.~Weghorn, and D.~Basin, ``Mining ABAC rules from sparse logs,'' in \emph{Proc. IEEE Eur. Symp. Secur. Privacy (EuroS\&P)}, 2018, pp.~31--46, doi: 10.1109/EuroSP.2018.00011.
\bibitem{karimi2022automatic} L.~Karimi, M.~Aldairi, J.~Joshi, and M.~Abdelhakim, ``An automatic attribute-based access control policy extraction from access logs,'' \emph{IEEE Trans. Dependable Secure Comput.}, vol.~19, no.~4, pp.~2304--2317, 2022, doi: 10.1109/TDSC.2021.3054331.
\bibitem{shang2024hierarchical} S.~Shang, X.~Wang, and A.~Liu, ``ABAC policy mining method based on hierarchical clustering and relationship extraction,'' \emph{Comput. Secur.}, vol.~139, 2024, Art. no.~103717, doi: 10.1016/j.cose.2024.103717.
\bibitem{guo2025onthefly} Y.~Guo, X.~Wang, M.~Yu, F.~Li, Z.~Pang, and L.~Fang, ``An on-the-fly framework for usable access control policy mining,'' \emph{Comput. Secur.}, vol.~150, 2025, Art. no.~104211, doi: 10.1016/j.cose.2024.104211.
\bibitem{aboukadri2024survey} S.~Aboukadri, A.~Ouaddah, and A.~Mezrioui, ``Machine learning in identity and access management systems: Survey and deep dive,'' \emph{Comput. Secur.}, vol.~139, 2024, Art. no.~103729, doi: 10.1016/j.cose.2024.103729.
\bibitem{nobi2025sok} M.~N. Nobi, M.~Gupta, R.~Krishnan, M.~S. Rana, L.~Praharaj, and M.~Abdelsalam, ``Machine learning in access control: A taxonomy,'' in \emph{Proc. 30th ACM Symp. Access Control Models Technol. (SACMAT)}, 2025, pp.~145--156, doi: 10.1145/3734436.3734453.
\bibitem{nobi2022deep} M.~N. Nobi, R.~Krishnan, Y.~Huang, M.~Shakarami, and R.~Sandhu, ``Toward deep learning based access control,'' in \emph{Proc. 12th ACM Conf. Data Appl. Secur. Privacy (CODASPY)}, 2022, pp.~143--154, doi: 10.1145/3508398.3511497.
\bibitem{cappelletti2019quality} L.~Cappelletti, S.~Valtolina, G.~Valentini, M.~Mesiti, and E.~Bertino, ``On the quality of classification models for inferring ABAC policies from access logs,'' in \emph{Proc. IEEE Int. Conf. Big Data}, 2019, pp.~4000--4007, doi: 10.1109/BigData47090.2019.9005959.
\bibitem{guo2017calibration} C.~Guo, G.~Pleiss, Y.~Sun, and K.~Q. Weinberger, ``On calibration of modern neural networks,'' in \emph{Proc. 34th Int. Conf. Mach. Learn. (ICML)}, 2017, pp.~1321--1330.
\bibitem{biggio2018wild} B.~Biggio and F.~Roli, ``Wild patterns: Ten years after the rise of adversarial machine learning,'' \emph{Pattern Recognit.}, vol.~84, pp.~317--331, 2018.
\bibitem{llamas2025adversarial} J.~Mart\'{\i}nez~Llamas, M.~Van~Hoof, D.~Preuveneers, and W.~Joosen, ``Adversarial robustness of machine learning-based access control,'' in \emph{Availability, Reliability and Security (ARES 2025)}, ser. LNCS, vol.~15994. Cham, Switzerland: Springer, 2025, pp.~88--105, doi: 10.1007/978-3-032-00630-1\_6.
\bibitem{anderson1972} J.~P. Anderson, ``Computer security technology planning study,'' USAF Electron. Syst. Div., Bedford, MA, USA, Tech. Rep. ESD-TR-73-51, 1972.
\bibitem{schneider2000enforceable} F.~B. Schneider, ``Enforceable security policies,'' \emph{ACM Trans. Inf. Syst. Secur.}, vol.~3, no.~1, pp.~30--50, 2000.
\bibitem{ligatti2009runtime} J.~Ligatti, L.~Bauer, and D.~Walker, ``Run-time enforcement of nonsafety policies,'' \emph{ACM Trans. Inf. Syst. Secur.}, vol.~12, no.~3, 2009, Art. no.~19.
\bibitem{bloem2015shield} R.~Bloem, B.~K\"{o}nighofer, R.~K\"{o}nighofer, and C.~Wang, ``Shield synthesis: Runtime enforcement for reactive systems,'' in \emph{Proc. Tools Algorithms Construction Anal. Syst. (TACAS)}, ser. LNCS, vol.~9035, 2015, pp.~533--548, doi: 10.1007/978-3-662-46681-0\_51.
\bibitem{alshiekh2018safe} M.~Alshiekh, R.~Bloem, R.~Ehlers, B.~K\"{o}nighofer, S.~Niekum, and U.~Topcu, ``Safe reinforcement learning via shielding,'' in \emph{Proc. AAAI Conf. Artif. Intell.}, vol.~32, no.~1, 2018, doi: 10.1609/aaai.v32i1.11797.
\bibitem{sha2001simplicity} L.~Sha, ``Using simplicity to control complexity,'' \emph{IEEE Softw.}, vol.~18, no.~4, pp.~20--28, 2001.
\bibitem{iyer2018mining} P.~Iyer and A.~Masoumzadeh, ``Mining positive and negative attribute-based access control policy rules,'' in \emph{Proc. 23rd ACM Symp. Access Control Models Technol. (SACMAT)}, 2018, pp.~161--172, doi: 10.1145/3205977.3205988.
\bibitem{iyer2022learning} P.~Iyer and A.~Masoumzadeh, ``Learning relationship-based access control policies from black-box systems,'' \emph{ACM Trans. Privacy Secur.}, vol.~25, no.~3, 2022, Art. no.~22, doi: 10.1145/3517121.
\bibitem{abujabal2020polisma} A.~Abu~Jabal, E.~Bertino, J.~Lobo, M.~Law, A.~Russo, S.~Calo, and D.~Verma, ``Polisma---A framework for learning attribute-based access control policies,'' in \emph{Computer Security -- ESORICS 2020}, ser. LNCS. Cham, Switzerland: Springer, 2020, pp.~523--544, doi: 10.1007/978-3-030-58951-6\_26.
\bibitem{zhang2026domain} S.~Zhang and P.~W.~L. Fong, ``Mining domain-based policies from massive and noisy access logs,'' in \emph{Proc. 31st ACM Symp. Access Control Models Technol. (SACMAT)}, 2026, pp.~2--13, doi: 10.1145/3750555.3811887.
\bibitem{liu1999msapriori} B.~Liu, W.~Hsu, and Y.~Ma, ``Mining association rules with multiple minimum supports,'' in \emph{Proc. 5th ACM SIGKDD Int. Conf. Knowl. Discovery Data Mining (KDD)}, 1999, pp.~337--341, doi: 10.1145/312129.312274.
\bibitem{weiss2004rarity} G.~M. Weiss, ``Mining with rarity: A unifying framework,'' \emph{ACM SIGKDD Explor. Newsl.}, vol.~6, no.~1, pp.~7--19, 2004, doi: 10.1145/1007730.1007734.
\bibitem{fisler2005margrave} K.~Fisler, S.~Krishnamurthi, L.~A. Meyerovich, and M.~C. Tschantz, ``Verification and change-impact analysis of access-control policies,'' in \emph{Proc. 27th Int. Conf. Softw. Eng. (ICSE)}, 2005, pp.~196--205, doi: 10.1145/1062455.1062502.
\bibitem{hughes2008xacml} G.~Hughes and T.~Bultan, ``Automated verification of access control policies using a SAT solver,'' \emph{Int. J. Softw. Tools Technol. Transf.}, vol.~10, no.~6, pp.~503--520, 2008, doi: 10.1007/s10009-008-0087-9.
\bibitem{turkmen2017smt} F.~Turkmen, J.~den~Hartog, S.~Ranise, and N.~Zannone, ``Formal analysis of XACML policies using SMT,'' \emph{Comput. Secur.}, vol.~66, pp.~185--203, 2017.
\bibitem{backes2018zelkova} J.~Backes \emph{et al.}, ``Semantic-based automated reasoning for AWS access policies using SMT,'' in \emph{Proc. Formal Methods Comput.-Aided Des. (FMCAD)}, 2018, pp.~1--9.
\bibitem{eiers2022quacky} W.~Eiers, G.~Sankaran, A.~Li, E.~O'Mahony, B.~Prince, and T.~Bultan, ``Quacky: Quantitative access control permissiveness analyzer,'' in \emph{Proc. 37th IEEE/ACM Int. Conf. Automated Softw. Eng. (ASE)}, 2022, pp.~1--5, doi: 10.1145/3551349.3559530.
\bibitem{cutler2024cedar} J.~W. Cutler \emph{et al.}, ``Cedar: A new language for expressive, fast, safe, and analyzable authorization,'' \emph{Proc. ACM Program. Lang.}, vol.~8, no.~OOPSLA1, pp.~670--697, 2024, doi: 10.1145/3649835.
\bibitem{iyer2022expressive} P.~Iyer, A.~Masoumzadeh, and P.~Narendran, ``On the expressive power of negated conditions and negative authorizations in access control models,'' \emph{Comput. Secur.}, vol.~116, 2022, Art. no.~102586, doi: 10.1016/j.cose.2021.102586.
\bibitem{alohaly2019nlp} M.~Alohaly, H.~Takabi, and E.~Blanco, ``Automated extraction of attributes from natural language attribute-based access control (ABAC) policies,'' \emph{Cybersecurity}, vol.~2, no.~1, 2019, doi: 10.1186/s42400-018-0019-2.
\bibitem{abdelgawad2023synthesizing} M.~Abdelgawad, I.~Ray, S.~Alqurashi, V.~Venkatesha, and H.~Shirazi, ``Synthesizing and analyzing attribute-based access control model generated from natural language policy statements,'' in \emph{Proc. 28th ACM Symp. Access Control Models Technol. (SACMAT)}, 2023, pp.~91--98, doi: 10.1145/3589608.3593844.
\bibitem{yang2025llmabac} M.~Yang, V.~Atluri, S.~Sural, and A.~Kundu, ``Extraction of machine enforceable ABAC policies from natural language text using LLM knowledge distillation,'' in \emph{Proc. 30th ACM Symp. Access Control Models Technol. (SACMAT)}, 2025, pp.~157--168, doi: 10.1145/3734436.3734447.
\bibitem{vatsa2025synthesizing} A.~Vatsa, P.~Patel, and W.~Eiers, ``Synthesizing access control policies using large language models,'' in \emph{Proc. IEEE/ACM Int. Workshop Natural Language-Based Softw. Eng. (NLBSE)}, 2025, pp.~13--16, doi: 10.1109/NLBSE66842.2025.00008.
\bibitem{bui2026llmmining} T.~Bui and E.~Martinez~Mota, ``Mining attribute-based access control policies with large language models,'' in \emph{Proc. 31st ACM Symp. Access Control Models Technol. (SACMAT)}, 2026, pp.~45--51, doi: 10.1145/3750555.3811907.
\bibitem{alhammuri2023zerotrust} K.~Al-hammuri, F.~Gebali, A.~Kanan, M.~Mamun, S.~M. Hazratifard, and H.~Alfar, ``Zero trust context-aware access control framework for IoT devices in healthcare cloud AI ecosystem,'' TechRxiv preprint, 2023, doi: 10.36227/techrxiv.23716221.
\bibitem{buccafurri2024enforcing} F.~Buccafurri, V.~De~Angelis, S.~Lazzaro, and A.~Pugliese, ``Enforcing security policies on interacting authentication systems,'' \emph{Comput. Secur.}, vol.~140, 2024, Art. no.~103771, doi: 10.1016/j.cose.2024.103771.
\bibitem{wang2025opmonitor} X.~Wang, Y.~Guo, Z.~Sun, M.~Yu, F.~Li, and L.~Fang, ``OPMonitor: Continuously monitoring residual over-granted permissions in verified access control policies,'' \emph{Comput. Secur.}, vol.~158, 2025, Art. no.~104623, doi: 10.1016/j.cose.2025.104623.
\bibitem{madkaikar2025queuing} G.~Madkaikar, K.~S.~M. Yelisetty, S.~Sural, J.~Vaidya, and V.~Atluri, ``Performance analysis of dynamic ABAC systems using a queuing theoretic framework,'' \emph{Comput. Secur.}, vol.~154, 2025, Art. no.~104432, doi: 10.1016/j.cose.2025.104432.
\bibitem{chow1970} C.~K. Chow, ``On optimum recognition error and reject tradeoff,'' \emph{IEEE Trans. Inf. Theory}, vol.~16, no.~1, pp.~41--46, 1970.
\bibitem{geifman2017selective} Y.~Geifman and R.~El-Yaniv, ``Selective classification for deep neural networks,'' in \emph{Proc. Adv. Neural Inf. Process. Syst. (NeurIPS)}, 2017, pp.~4878--4887.
\bibitem{vovk2005alrw} V.~Vovk, A.~Gammerman, and G.~Shafer, \emph{Algorithmic Learning in a Random World}. New York, NY, USA: Springer, 2005.
\bibitem{shafer2008tutorial} G.~Shafer and V.~Vovk, ``A tutorial on conformal prediction,'' \emph{J. Mach. Learn. Res.}, vol.~9, pp.~371--421, 2008.
\bibitem{vovk2012conditional} V.~Vovk, ``Conditional validity of inductive conformal predictors,'' in \emph{Proc. Asian Conf. Mach. Learn. (ACML)}, ser. PMLR, vol.~25, 2012, pp.~475--490.
\bibitem{angelopoulos2023gentle} A.~N. Angelopoulos and S.~Bates, ``Conformal prediction: A gentle introduction,'' \emph{Found. Trends Mach. Learn.}, vol.~16, no.~4, pp.~494--591, 2023.
\bibitem{angelopoulos2024crc} A.~N. Angelopoulos, S.~Bates, A.~Fisch, L.~Lei, and T.~Schuster, ``Conformal risk control,'' in \emph{Proc. Int. Conf. Learn. Represent. (ICLR)}, 2024.
\bibitem{natarajan2013noisy} N.~Natarajan, I.~S. Dhillon, P.~Ravikumar, and A.~Tewari, ``Learning with noisy labels,'' in \emph{Proc. Adv. Neural Inf. Process. Syst. (NIPS)}, 2013, pp.~1196--1204.
\bibitem{frenay2014survey} B.~Fr\'{e}nay and M.~Verleysen, ``Classification in the presence of label noise: A survey,'' \emph{IEEE Trans. Neural Netw. Learn. Syst.}, vol.~25, no.~5, pp.~845--869, 2014.
\bibitem{bui2025abaclab} T.~Bui, A.~Matricia, E.~Contreras, R.~Mauvais, L.~Medina, and I.~Serrano, ``ABAC Lab: An interactive platform for attribute-based access control policy analysis, tools, and datasets,'' in \emph{Proc. 30th ACM Symp. Access Control Models Technol. (SACMAT)}, 2025, pp.~111--116, doi: 10.1145/3734436.3734441. [Online]. Available: \url{https://github.com/ABAC-Lab-Admin/ABAC-Lab}
\bibitem{decat2014edocument} M.~Decat, J.~Bogaerts, B.~Lagaisse, and W.~Joosen, ``The e-document case study: Functional analysis and access control requirements,'' Dept. Comput. Sci., KU Leuven, Leuven, Belgium, CW Rep. CW654, 2014.
\bibitem{brucker2009breakglass} A.~D. Brucker and H.~Petritsch, ``Extending access control models with break-glass,'' in \emph{Proc. 14th ACM Symp. Access Control Models Technol. (SACMAT)}, 2009, pp.~197--206, doi: 10.1145/1542207.1542239.
\bibitem{rostad2006audit} L.~R{\o}stad and O.~Edsberg, ``A study of access control requirements for healthcare systems based on audit trails from access logs,'' in \emph{Proc. 22nd Annu. Comput. Secur. Appl. Conf. (ACSAC)}, 2006, pp.~175--186, doi: 10.1109/ACSAC.2006.8.
\end{thebibliography}
\end{document}